\documentclass[twocolumn,10pt]{IEEEtran}
\normalsize
\usepackage{bbm}
\usepackage[dvips]{graphicx}
\usepackage{adjustbox}
\usepackage{enumitem}
\usepackage{cite,algorithm,algorithmic,amsmath,amssymb,amsthm,empheq,mhsetup}
\usepackage{subfigure,amsfonts,balance}
\usepackage{epstopdf}
\usepackage{setspace}
\usepackage[dvipsnames]{xcolor}

\DeclareMathOperator{\EEE}{\mathbb{E}}

\DeclareMathOperator{\f}{\pmb{f}}
\DeclareMathOperator{\aaa}{\pmb{a}}
\DeclareMathOperator{\FF}{\mathcal{F}}

\DeclareMathOperator{\K}{\mathcal{K}}

\DeclareMathOperator{\OO}{\mathcal{O}}

\DeclareMathOperator{\vv}{\pmb{v}}

\DeclareMathOperator{\z}{\pmb{z}}

\DeclareMathOperator{\HHH}{\mathcal{H}}

\DeclareMathOperator{\RRR}{\mathbb{R}}
\DeclareMathOperator{\RR}{\pmb{R}}
\DeclareMathOperator{\rrr}{\pmb{r}}

\DeclareMathOperator{\G}{\pmb{G}}

\DeclareMathOperator{\LL}{\mathcal{L}}

\DeclareMathOperator{\CN}{\mathcal{CN}}

\DeclareMathOperator{\bb}{\pmb{b}}

\DeclareMathOperator{\NN}{\mathcal{N}}

\DeclareMathOperator{\e}{\pmb{e}}

\DeclareMathOperator{\rr}{\pmb{r}}
\DeclareMathOperator{\x}{\pmb{x}}
\DeclareMathOperator{\ttt}{\pmb{t}}

\DeclareMathOperator{\y}{\pmb{y}}

\DeclareMathOperator{\Ss}{\pmb{S}}

\DeclareMathOperator{\uu}{\pmb{u}}

\DeclareMathOperator{\g}{\pmb{g}}

\DeclareMathOperator{\OOmega}{\pmb{\omega}}
\DeclareMathOperator{\ETA}{\pmb{\eta}}

\DeclareMathOperator{\ZETA}{\pmb{\zeta}}

\newtheorem{remark}{Remark}
\newtheorem{proposition}{Proposition}

\ifCLASSINFOpdf
\else
\fi

\begin{document}
	\bstctlcite{IEEEexample:BSTcontrol}
	%
	\title{\huge Energy-Efficient Massive MIMO for Federated Learning: Transmission Designs and Resource Allocations \vspace{-0mm}}
	%
	%
	\author{Tung~T.~Vu,~\IEEEmembership{Member,~IEEE}, 
		Hien~Quoc~Ngo,~\IEEEmembership{Senior~Member,~IEEE}, 
		Minh~N.~Dao, Duy~T.~Ngo,~\IEEEmembership{Senior~Member,~IEEE}, 
		Erik~G.~Larsson,~\IEEEmembership{Fellow,~IEEE},
		and Tho~Le-Ngoc,~\IEEEmembership{Life Fellow,~IEEE}
		\thanks{The work of T.~T.~Vu and H.~Q.~Ngo was supported by the U.K. Research and Innovation Future Leaders Fellowships under Grant MR/S017666/1. The work of M.~N.~Dao benefited from the FMJH Program Gaspard Monge for optimization and operations research and their interactions with data science, and was supported by a public grant as part of the Investissement d'avenir project, reference ANR-11-LABX-0056-LMH, LabEx LMH. The work of T.~T.~Vu and E.~G.~Larsson was
			supported in part by ELLIIT and in part by the KAW Foundation. The work of Tho Le-Ngoc was supported by the Natural Sciences and Engineering Research Council of Canada.}
		\thanks{T.~T.~Vu and H.~Q.~Ngo are with the Institute of Electronics, Communications, and Information Technology (ECIT), Queen's University Belfast, Belfast BT3 9DT, UK (e-mail: \{t.vu, hien.ngo\}@qub.ac.uk).}
		\thanks{M.~N.~Dao is with the School of Science, RMIT University, Melbourne, VIC 3000, Australia (e-mail: minh.dao@rmit.edu.au).}
		\thanks{D.~T.~Ngo is with the School of Engineering, The University of Newcastle, Callaghan, NSW 2308, Australia (e-mail: duy.ngo@newcastle.edu.au).}
		\thanks{T.~T.~Vu and E. G. Larsson are with the Department of Electrical Engineering (ISY), Link\"{o}ping University, SE-581 83 Link\"{o}ping, Sweden (e-mail: \{thanh.tung.vu, erik.g.larsson\}@liu.se).}
		\thanks{T. Le-Ngoc is with the Department of Electrical and Computer Engineering, McGill University, Montreal, Canada (e-mail: tho.le-ngoc@mcgill.ca)}
		\vspace{-5mm}
	}
	
		
		\maketitle
		\allowdisplaybreaks
		\begin{spacing}{1}
			\begin{abstract}
				This work proposes novel synchronous, asynchronous, and session-based designs for energy-efficient massive multiple-input multiple-output networks to support federated learning (FL). 
				The synchronous design relies on strict synchronization among users when executing each FL communication round, while the asynchronous design allows more flexibility for users to save energy by using lower computing frequencies. 
				The session-based design splits the downlink and uplink phases in each  FL communication round into separate sessions. 
				In this design, we assign users such that one of the participating users in each session finishes its transmission and does not join the next session. 
				As such, more power and degrees of freedom will be allocated to unfinished users, resulting in higher rates, lower transmission times, and hence, higher energy efficiency. In all three designs, we use zero-forcing processing for both uplink and downlink, and develop algorithms that optimize user assignment, time allocation, power, and computing frequencies to minimize the energy consumption at the base station and users, while guaranteeing a predefined maximum execution time of each FL communication round.
			\end{abstract}
		\end{spacing}
		
		\vspace{-0mm}

		%
		\IEEEpeerreviewmaketitle
		
		\vspace{-0mm}
		\section{Introduction}
		\vspace{-0mm}
		\label{sec:Introd}
		Over the past few decades, communication systems with the Internet and mobile telephony  brought much convenience to human life \cite{ji21CSM,walid20N,yang19N}. Recently, the rapid development of artificial intelligence 
		has contributed  to the modernization of our world with a wide range of applications such as smart cities and autonomous cars \cite{viet21A,song20CIM,lahmeri21OJCS
		}. However, current communication systems are also facing big challenges. Specifically, since users (UEs) need to send their data over a shared medium, their data privacy can be compromised, as already happened  \cite{jere21SP
		}. 
		At the same time, mobile data traffic is anticipated to increase dramatically during 2020--26, at up to $32\%$ per month \cite{Ericsion21}. This in turn has led to concerns about energy consumption and carbon emissions, where communication systems are projected to contribute significantly \cite{su18}.  
		On the other hand, according to the report \cite{Ericsion18RB}, the information and communication technology sector was estimated to account for a portion of $1.4\%$ of global carbon emissions in $2015$. More importantly, this portion is likely to grow in the future when the number of internet-of-things devices grows exponentially. 
		Therefore, it is critical for future communication systems not only to be integrated with machine learning applications, but also to preserve privacy and be energy-efficient.
		
		Federated learning (FL) is a distributed learning framework that offers high  privacy  and communication efficiency 
		\cite{khan21CST,savazzi21CM,niknam20CM,liu20CC
		}. Especially, in FL, no raw data are shared during the learning process. 
		An FL process is jointly implemented by several UEs and a central server. 
		First, the central server sends a global model update to all the UEs. 
		Each UE uses this  model update, along with its private training data, to compute its own local learning model update. The UEs then send their local updates back to the central server for updating the global model update. This process is repeated until a certain level of learning accuracy is reached. Here, since the size of the model updates sent over the network is much smaller than that of the raw data, communication efficiency is much improved.

		
		\subsection{Review of Related Literature}
		In the literature, there are only several works that study energy-efficient implementations of FL over wireless networks, e.g., \cite{li21INFOCOM,kim22WCL,yue22JSAC,jin22TWC,yang21TWC,zeng21TWC,viet22TVT} and references therein. These papers can be categorized into learning-oriented and communication-oriented directions. The learning-oriented direction seeks learning solutions to reduce the energy consumed in the networks. In particular, \cite{li21INFOCOM} proposes an FL algorithm that adapts the compression parameters to minimize energy consumption at UEs. The work of \cite{kim22WCL} proposes a novel joint dataset and computation management scheme that trades off between learning accuracy and energy consumption for energy-efficient FL in mobile edge computing. Reference \cite{yue22JSAC} introduces a federated meta-learning algorithm together with a resource allocation scheme to jointly improve convergence rate and minimize energy cost. Finally, \cite{jin22TWC} develops a SignSGD-based FL algorithm where local processing and communication parameters are chosen to achieve a desired balance between learning performance and energy consumption.

		The communication-oriented direction does not propose new FL algorithms, but rather develops communication protocols and system designs to reduce the energy consumption of an  FL process run over a wireless network \cite{yang21TWC,zeng21TWC,chen21TWC,viet22TVT}. Compared to the learning-oriented direction, the communication-oriented gives more insights into how FL should be implemented at the physical layer. Specifically, \cite{yang21TWC} minimizes energy consumption at user devices by optimally allocating bandwidth, power, and computing frequency. Reference \cite{zeng21TWC} proposes another resource allocation algorithm for FL networks, in which each user is equipped with a CPU-GPU platform for heterogeneous computing. 
		The authors in \cite{chen21TWC} proposed a joint communication and learning framework that improves the learning performance while keeping the energy consumption acceptable on each user device.
		The work of \cite{viet22TVT} designs a network with unmanned aerial vehicles and wireless powered communications to provide an energy-efficient FL solution.

		\subsection{Research Gap and Main Contributions}
		The ongoing research efforts in the communication-oriented direction have mainly used frequency-division multiple  access (FDMA) to support FL. The drawback of FDMA networks is that the spectral and energy efficiencies are very low when the channel is shared by many users. It is therefore desirable to propose a novel network design to implement FL frameworks with a much higher energy efficiency.

		This research gap in the literature has motivated us to consider a massive multiple-input multiple-output (mMIMO) network to implement wireless FL in an energy-efficient manner.  The use of massive MIMO to support FL has been shown to be very efficient \cite{vu20TWC,ema22SPAWC,vu21SPAWC,wei22ICC,jeon21TWC,mu22TNSE}, compared to conventional FDMA or  time-division multiple access (TDMA) schemes. The main reasons for this are: (i) massive MIMO can simultaneously serve many users;  (ii) massive MIMO offers huge spectral efficiencies, and hence, can significantly reduce the training time; and (iii) massive MIMO provides high energy efficiency \cite{emil15TWC}. As a result, massive MIMO fits well with federated learning applications that require a large number of
		energy-efficient and low-latency transmissions between user devices and the server at the same time (e.g., a camera network of augmented reality users in the same cell building a model for object detection and classification, a vehicular network of clients equipped with various sensors building a model for image classification \cite{khan21CST}).
		
		The specific contributions of this paper are summarized as follows:
		\vspace{-0mm}
		\begin{itemize}[noitemsep,nolistsep]
			\item To support FL  over wireless networks, we propose to use mMIMO and let each FL communication round be executed within one large-scale coherence time\footnote{The large-scale coherence time is defined as the time interval during which the large-scale fading coefficients remain approximately constant.}. Owing to a high array gain and multiplexing gain, mMIMO can offer very high data rates to all UEs simultaneously in the same frequency band \cite{ngo16}. Therefore, it is expected to guarantee a stable operation during each communication round (and hence the whole FL process). 
			\item We introduce three novel transmission designs for the steps within one FL communication round. The downlink (DL) transmission, the computation at the UEs, and the uplink (UL) transmission, are implemented in a synchronous, asynchronous, or session-based manner. The synchronous design strictly synchronizes UEs in each step of one FL communication round.  The asynchronous design allows more flexibility for UEs to save energy by using lower computing frequencies. The session-based design splits the DL and UL steps into separate sessions. The UEs are then assigned such that one of the participating UEs in each session will complete its transmission and does not join subsequent sessions. This design allows more power to be allocated to the uncompleted UEs. This results in higher rates, lower transmission times, and higher energy efficiency. 
			In all three designs, both DL and UL transmissions use a dedicated pilot assignment scheme for channel estimation and zero-forcing (ZF) processing. 
			\item For each proposed transmission design, we formulate a problem of optimizing user assignment, time allocation, transmit power, and computing frequency to minimize the total energy consumption in each FL communication round, subject to a quality-of-service constraint. 
			The formulated problems are challenging due to their nonconvex and combinatorial (mixed-integer) nature. Existing solutions to problems in standard massive MIMO systems cannot be used in a straightforward manner to solve the formulated problems. As such, we propose novel algorithms that are proven to converge to stationary points, i.e., Fritz John and Karush–Kuhn–Tucker solutions, of the formulated problems.
			\item  We show by numerical results that our proposed designs significantly reduce the energy consumption per FL communication round compared to heuristic baseline schemes. The presented numerical results also confirm that the session-based design outperforms the synchronous and asynchronous designs.
		\end{itemize}
		\vspace{-0mm}
		

		It is noted that the idea of the proposed synchronous design is similar to the transmission scheme in \cite{vu20TWC}. However, the resource allocation algorithm in \cite{vu20TWC} for minimizing the FL training time in a cell-free massive MIMO network cannot be straightforwardly applied to solve the more complex problem of minimizing the energy consumption of an mMIMO network, as treated in this work. On the other hand, the proposed synchronous and asynchronous designs are different from those in \cite{vu21GC}. They use dedicated pilot assignment and ZF processing for each UE, while those in \cite{vu21GC} use co-pilot assignment and ZF processing for each group of UEs. These key distinctions result in major differences in the respective problem formulations and algorithms for resource allocation.
		
		\emph{Notation:}
		We use boldface symbols for vectors and capitalized boldface symbols for matrices.
		$\RRR^d$ denotes a space where its elements are real vectors of length $d$.
		$\pmb{X}^*$ and $\pmb{X}^H$ represent the conjugate and conjugate transpose of a matrix $\pmb{X}$, respectively.
		$\CN(\pmb{0},\pmb{Q})$ denotes the circularly symmetric complex Gaussian distribution with zero mean and covariance $\pmb{Q}$.
		$\EEE\{x\}$ denotes the expected value of a random variable $x$.

		\vspace{-0mm}
		\section{Novel Massive MIMO Designs to Support Federated Learning Networks}
		\label{sec:design}
		\vspace{-0mm}
		
		\begin{figure*}[t!]
			\centering
			\includegraphics[width=1\textwidth]{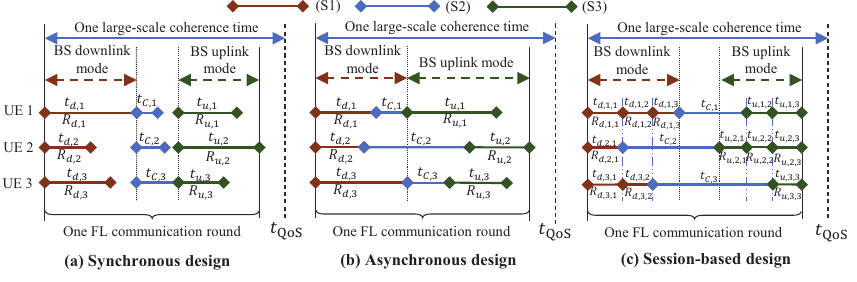}
			\vspace{-4mm}
			\caption{Illustration of one FL communication round over the considered mMIMO network with three UEs.}
			\vspace{-6mm}
			\label{fig:time1}
		\end{figure*}
		
		In this work, we focus on the optimization of communication resources in a massive MIMO wireless network that supports FL applications. Specifically, we consider the use of a \textit{standard} FL algorithm  and develop optimized transmission designs  that support this FL framework.
		We consider a network that supports FL algorithms with a synchronous aggregation mode\footnote{FL algorithms with the synchronous aggregation mode wait to receive all  local model updates sent from users before aggregation, while the FL algorithms with the asynchronous aggregation mode do not. The FL algorithms with synchronous aggregation normally outperforms the FL algorithms operating with asynchronous aggregation in terms of convergence rate and accuracy. Research on improvement of learning performance of the FL algorithms with asynchronous aggregation is still in its infancy, while FL algorithms with synchronous aggregation are well studied \cite{xu21}. Therefore, our paper focuses on transmission protocols supporting FL  with synchronous aggregation \cite{canh21TN, sai20ICML, canh20NIPS, felix20TNNLS, tran19INFOCOM, mcmahan17AISTATS}.}. 
		In general, such an FL network includes a group of UEs and a central server. Each FL communication round includes $K$ UEs and the following four basic steps \cite{canh21TN,sai20ICML,canh20NIPS,felix20TNNLS,tran19INFOCOM,mcmahan17AISTATS}:
		\begin{enumerate}[label={(S\arabic*)}]
			\item The central server sends a global update to the UEs.
			\item Each UE updates its local learning problem with the global update and its local data, and then computes a local update by solving the local problem.
			\item Each UE sends its local update to the central server.
			\item The central server recomputes the global update by aggregating the received local updates from all the UEs.
		\end{enumerate}
		The above process repeats until a certain level of learning accuracy is attained. 
		Details on the local and global updates along with their associated computations are thoroughly discussed in \cite{canh21TN,sai20ICML,canh20NIPS,felix20TNNLS,tran19INFOCOM,mcmahan17AISTATS}.  
		We assume that before our proposed schemes are undertaken, all the UEs that participate in each FL communication round have sufficient computational capabilities to update their models. This assumption is widely accepted in the literature on wireless network designs for supporting federated learning, e.g., \cite{yang21TWC,zeng21TWC,wu22TCOM,viet22TVT,pavlos22CL} and references therein.

			We note that the aggregation of  local updates of the UEs can be performed by two approaches. The first approach makes the aggregation in the digital domain \cite{yang21TWC,zeng21TWC,wu22TCOM,viet22TVT,pavlos22CL}, and is called DigComp. 
			The second approach leverages the signal superposition property to aggregate in the analog domain and is called over-the-air computation (AirComp) \cite{yang20TWC,amiri21TWC,ema22SPAWC,wei22arxiv,fan22TWC,cao22JSAC,shao22TWC}. 
			While DigComp leverages the capability of traditional digital transmission in  wireless systems that are deployed and standardized, AirComp is an emerging approach which is still under basic development and not yet supported by cellular systems \cite{henrik22}. 
			Most existing works using AirComp require the UEs to acquire CSI, which in itself is a very challenging task. Research on wireless network designs using AirComp without CSI acquisition is still in its infancy  \cite{amiri21TWC,ema22SPAWC,wei22arxiv}.   
			In this work, we follow the DigComp approach and propose energy-efficient transmission designs for  massive MIMO systems to support FL. The topic of using the AirComp approach for energy-efficient transmission designs to support FL is left for future work.

		\vspace{-0mm}
		\subsection{Proposed Transmission Designs to Support Federated Learning Networks}
		\label{transmission:designs}
		To support FL in the network, we propose to use mMIMO technology where the BS acts as the central server. Accordingly,  Steps (S1) and (S3) of each FL communication round take place over the DL and UL of the mMIMO system, respectively. 
		Each FL communication round is assumed to be executed within a channel large-scale coherence time, which is a reasonable assumption for typical network scenarios \cite{vu21SPAWC,vu21ICC,vu20TWC}. Under this assumption, we propose the following transmission schemes \footnote{UE selection could be beneficial for improving the energy efficiency of the system, especially in the case that some UEs have very bad channel conditions. However, UE selection reduces the number of UEs that participate in the FL process, and hence, would affect the FL performance (i.e., test accuracy) \cite{yang20TWC}.  Since we mainly focus on the communication aspects in a standard FL framework, we do not incorporate the UE selection process into our proposed \textit{transmission designs}, but assume that all $K$ UEs participate in each FL communication round. This assumption is made in much of the literature on  wireless network design for support of federated learning, e.g., \cite{yang21TWC,zeng21TWC,wu22TCOM,viet22TVT,pavlos22CL}. More importantly, although we do not take into account the UE selection part in the transmission designs, our proposed transmission schemes can still be used to support FL frameworks that have UE selection in their FL algorithms. Specifically, in each communication round of such FL algorithms, different values of $K$ and different UEs can be selected from a larger pool of UEs using the UE selection scheme in the FL algorithm. Then, our optimization problems can be reformulated for the given new $K$ UEs without any changes in their mathematical structure.}
		to support Steps (S1)--(S3) of each FL communication round:
		\begin{itemize}
			\item[1)] \textbf{Synchronous Design}: As shown in Fig.~\ref{fig:time1}(a), the synchronous design requires a  certain degree of synchronization among the UEs when executing the steps of one FL communication round. In particular, the UEs are synchronized for steps (S2) and (S3) to start simultaneously at all UEs. The UEs' rates are taken to be the achievable rates when all the UEs' transmissions are being active. 
			\item[2)] \textbf{Asynchronous Design}: Compared with the synchronous design, the asynchronous design uses the same rate assignment scheme. The DL (UL) rate of each user is kept fixed for the whole DL (UL) mode. However, the asynchronous design has a different transmission protocol. The asynchronous design only requires the UEs to start Step (S1) simultaneously. As shown in Fig.~\ref{fig:time1}(b), UEs have more flexibility in executing Steps (S1)--(S3). This is because they can transmit their local model updates in Step (S3) immediately after they complete Step (S2), as long as their UL transmission is performed during the BS UL mode. Thus, the UEs in the asynchronous design need not wait for other UEs, as is the case of the synchronous design. Instead, they can use the waiting time to compute their local model updates with a lower clock frequency to save energy. Also, thanks to the flexible synchronization requirement among the UEs, the asynchronous design has a significantly lower signalling overhead compared to the synchronous design, especially when the number of UEs is large.  
			\item[3)] \textbf{Session-based Design}: In the asynchronous design, the DL (UL) rate of each user is kept fixed for the whole DL (UL) duration. This is not efficient because, for each mode, after some time, some users may complete their transmissions. Hence, other users can increase their rates owing to the reduced level of interference and increased availability of power (on DL).
			Based on this observation, we propose the session-based design in Fig.~\ref{fig:time1}(c). Here, instead of using one single session for each step (S1) or (S3), we use multiple sessions to serve UEs in steps (S1) and (S3). After each session, one user completes its transmission, and the rates of other users are adapted accordingly. Since there are fewer UEs competing for power in each session, more power can be allocated to the UEs that have not yet completed their transmissions. In addition, the inter-user interference reduces, which leads to higher rates, faster transmission and better energy efficiency compared to the other designs. 
		\end{itemize}

			

		\vspace{-0mm}
		\section{System Models}
		\vspace{-0mm}
		\label{sec:SystModel}
		This section provides detailed system models for the proposed designs. 
		As discussed in Section~\ref{sec:design}, because the synchronous and asynchronous designs use the same rate assignment scheme, their system models are similar. On the other hand, as can be seen from Fig.~\ref{fig:time1}, the asynchronous design is a special case of the session-based design with a single session
		. Based on these observations, the system model of the session-based design is therefore provided as a general model, followed by the specific models for the asynchronous and synchronous designs. 
		
		In the considered mMIMO model, a BS equipped with $M$ antennas serves $K$ UEs each equipped with a single antenna at the same time and in the same frequency bands, using time-division duplexing.
		The channel vector from a UE $k$ to the BS is denoted by $\g_{k} \!=\! (\beta_{k})^{1/2}\tilde{\g}_{k},$
		where $\beta_{k}$ and $\tilde{\g}_{k} \sim \CN(\pmb{0},\pmb{I}_M)$ 
		are the corresponding large-scale fading coefficient and small-scale fading coefficient vector, respectively.
		In this work, we consider low mobility scenarios with a large coherence interval $\tau_c$.
		Each FL communication round is executed in one large-scale coherence time \cite{vu20TWC} (see Fig.~\ref{fig:FLround})
		The DL transmission for the global update in Step (S1) and the UL transmission for the local update in Step (S3) span multiple (small-scale) coherence times. 
		
		\vspace{-0mm}
		\subsubsection{Step (S1)} 
		The BS sends the parameter vector to all the UEs in $K$ sessions. Each coherence block of this step involves two phases: UL channel estimation and DL payload data transmission.
		Define an indicator $a_{k,i}$ as 
		\begin{align}\label{a}
			a_{k,i} \triangleq
			\begin{cases}
				1,& \text{if the BS serves a UE $k$ in a session $i$,}\\
				0, & \mbox{otherwise}.
			\end{cases}
		\end{align}
		Let $\K_i\triangleq\{k|a_{k,i} = 1\}$ be the set of $K_i=\sum_{k\in\K}a_{k,i}$ UEs participating in a session $i\in\K$. We have
		\begin{align}
			\label{sumaki}
			& a_{k,1} = 1, \sum_{k\in\K} a_{k,i} = K-i+1, a_{k,i} \leq a_{k,i-1},  \forall k,i
		\end{align}
		to make sure that all the UEs are served in session $1$. Also, in each of the subsequent sessions, one UE is instructed to finish its transmission such that it does not join the next sessions.
		Doing this helps the UEs who are yet to finish their transmissions in that they get assigned more power and experience a lower level of inter-user interference, which translates into higher data rates.
		
		Here, the asynchronous and synchronous designs are considered as the same special case 
		when all the UEs are served in a single session,
		$i=1$, and $a_{k,1}=1, \forall k$.

		\textbf{UL channel estimation}:
		For each coherence block of length $\tau_c$, each UE sends its dedicated pilot of length $\tau_{d,p}$ to the BS. 
		We assume that the pilots of all UEs are pairwisely orthogonal,
		which requires $\tau_{d,p}\geq K$.\footnote{It is possible to let only the participating UEs send their pilots in session $i$ and choose $\tau_{d,p} = K_i$ to increase the remaining interval for payload data transmission. However, since $\tau_c$ is normally much larger than $K\geq K_i$  \cite{vu20TWC,ngo17TWC}, letting all UEs send their pilots, i.e., choosing $\tau_{d,p} = K$, has a negligible impact on data rates. In addition, using a pilot length $\tau_{d,p} = K>K_i$ makes the channel estimation better than with $\tau_{d,p} = K_i$, and hence, can potentially improve the data rates.} 
		At the BS, the channel $\g_{k}$ between a UE $k$ and the BS is estimated by using the received pilots and  minimum mean-square error (MMSE) estimation. The MMSE estimate $\hat{\g}_{k}$ of $\g_{k}$ is
		distributed as $\CN(\pmb{0},\hat{\sigma}_{k}^2\pmb{I}_M)$, where $\hat{\sigma}_k^2 = \frac{\tau_{d,p} \rho_{p} \beta_k^2 }{ \tau_{d,p} \rho_{p} \beta_k +1 }$ and $\rho_{p}$ is the normalized transmit power of each pilot symbol \cite[(3.8)]{ngo16}.
		We also denote by $\hat{\G}_i\triangleq [\dots,\hat{\g}_{k},\dots], \forall k\in\K_i$, the matrix obtained by stacking the channels of all participated UEs in a session $i$.

		\textbf{DL payload data transmission}: We assume that the BS uses a unicast scheme and ZF precoding
		to transmit the global training update to the $K$ UEs. Let $s_{d,k,i}$, where $\EEE\{|s_{d,k,i}|^2\}=1$, be the data symbol intended for a UE $k$ in a session $i$. With ZF, the signal transmitted by the BS in the session $i$ is 
		$\vv_{d,i}\!=\! \sqrt{\rho_{d}}\sum_{\ell\in \K_i}\sqrt{\eta_{k,i}}\uu_{\ell,i} s_{d,\ell,i},$
		where $\uu_{k,i}\! =\! \sqrt{\hat{\sigma}_{k}^2(M\!-\!K_i)} \hat{\G}_i(\hat{\G}_i^H\hat{\G}_i)^{-1}\e_{k,K_i}$ is the ZF precoding vector,
		$\eta_{k,i}$ is a power control coefficient, $\e_{k,K_i}$ is the $k$-th column of $\pmb{I}_{K_i}$, and $\rho_{d}$ is the maximum normalized transmit power at the BS. Note that ZF  requires $M \geq K_i$.
		The transmitted power at the BS must meet the average normalized power constraint  $\EEE\{|\vv_{d,i}|^2\}\leq \rho_d$, which can also be expressed as 
		\begin{align}
			\label{powerdupperbound}
			\sum_{k\in\K_i}\eta_{k,i} \leq 1, \forall i\in\K.
		\end{align}
		Here, we have 
		\begin{align}
			\label{eta-a-relation} 
			(\eta_{k,i} = 0, \text{if $a_{k,i} = 0$}), \forall k,i
		\end{align}
		to ensure that no power is allocated to the UEs that are not served in the session $i$. 
		The achievable rate of the UE $k$ in the session $i$ is given by $R_{d,k,i}(\ETA_i) = \frac{\tau_c - \tau_{d,p}}{\tau_c}B\log_2 \big( 1 + \text{SINR}_{d,k,i}(\ETA_i)\big)$, 
		where $B$ is the transmission bandwidth, $\ETA_i\triangleq\{\eta_{k,i}\}_{k\in\K}$, and
			$\text{SINR}_{d,k,i}(\ETA_i) =
			\frac{(M-K_i)\rho_d \hat{\sigma}_k^2\eta_{k,i}}
			{\rho_d (\beta_k - \hat{\sigma}_k^2) \sum_{\ell\in\K_i} \eta_{\ell,i} +1} 
			\overset{\eqref{eta-a-relation}}{=} \frac{(M-K_i)\rho_d \hat{\sigma}_k^2\eta_{k,i}}
			{\rho_d (\beta_k - \hat{\sigma}_k^2) \sum_{\ell\in\K} \eta_{\ell,i} +1}$
		is the effective DL signal-to-interference-plus-noise ratio (SINR) \cite[(3.56)]{ngo16}.
		
		Similarly, the power constraint at the BS and the achievable rate at the UE $k$ in the asynchronous and synchronous designs are given as
		\begin{align}
			\label{powerdupperbound:asyn:syn}
			&\sum_{k\in\K}\eta_{k} \leq 1
			\\
			\label{Rd}
			&R_{d,k}(\ETA) = \frac{\tau_c - \tau_{d,p}}{\tau_c}B\log_2 \big( 1 + \text{SINR}_{d,k}(\ETA)\big),
		\end{align}
		where $\ETA\triangleq\{\eta_{k}\}_{k\in\K}$ are the power control  coefficients, and $\text{SINR}_{d,k}(\ETA) =
		\frac{(M-K)\rho_d \hat{\sigma}_k^2\eta_{k}}
		{\rho_d (\beta_k - \hat{\sigma}_k^2) \sum_{\ell\in\K} \eta_{\ell} +1}$. 
		
		\begin{figure}[t!]
			\centering
			\includegraphics[width=0.5\textwidth]{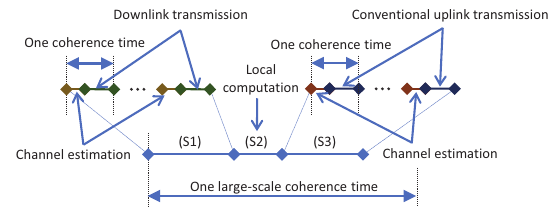}
			\vspace{-0mm}
			\caption{{Operation of one FL communication round in the considered massive MIMO network.}}
			\label{fig:FLround}
		\end{figure}
		
		\vspace{-0mm}
		\begin{remark}
			The same global training update can be coded differently for different UEs to improve the spectral efficiency of the DL transmission. 
			Specifically, it can be transmitted by either a multicast scheme or a unicast scheme \cite{sadeghi18TWC}. As shown in Fig.~5 of \cite{sadeghi18TWC}, in a massive MIMO system where the same message is sent to all users, the scheme using unicast and ZF is recommended in almost all cases, except when the coherence interval is short (small $\tau_c$) or  the number of antennas $M$ at the BS is small. In our paper, we consider low mobility scenarios (i.e., large $\tau_c$) with a large value of $M$.
			Therefore, we  choose unicast and ZF precoding for our transmission scheme. We verify the advantage of this choice over the multicast scheme by Fig.~\ref{Fig:UCvsMC}, which compares the unicast scheme and the multicast schemes in a single group of UEs. From the figure, in terms of per-UE rates, the unicast with dedicated pilots significantly outperforms the multicast counterparts in both dedicated and co-pilot pilot designs. 
			We also note that the difference in the global training update for each user is the difference of the symbols that encode the same global training update for different users. There is no change in the FL model of the standard FL framework discussed in Section II.~A. On the other hand, ZF precoding, while simple, performs very closely to the optimal precoding in massive MIMO \cite{ngo16,emil17}. That is why ZF precoding is employed in this paper, to achieve both simplicity and good performance.
		\end{remark}
		
		\begin{figure}[t!]
			\centering
			{\includegraphics[width=0.45\textwidth]{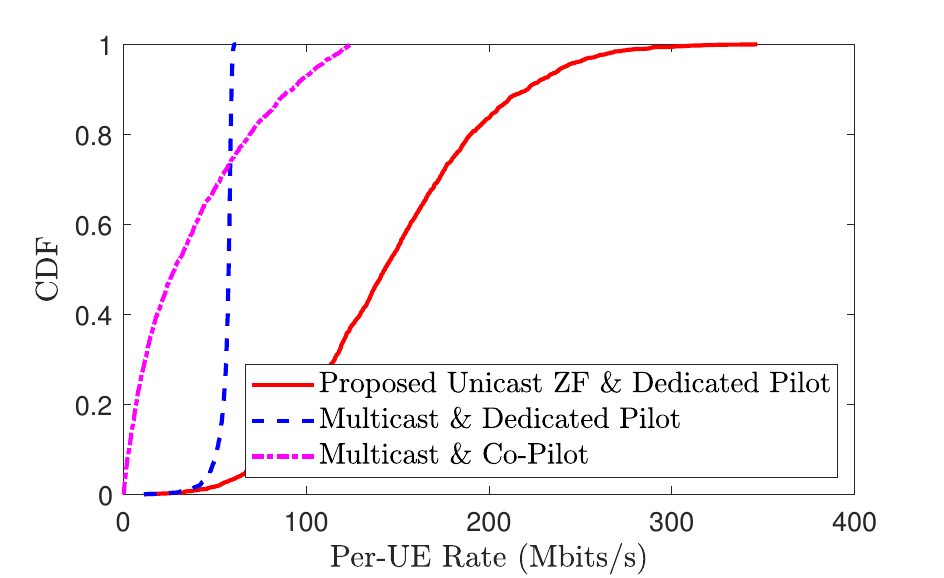}}
			\vspace{-0mm}
			\caption{Comparison of downlink per-UE rates in a single group of UEs between the unicast ZF scheme with a dedicated pilot design and the multicast schemes with dedicated and co-pilot pilot designs. Here, $M=75, K=10, \eta_k = 1/K, \forall k$. All other parameters are the same as  in our simulation results (see Section~\ref{sec:sim:parameter}).}
			\label{Fig:UCvsMC}
		\end{figure}

		\vspace{-0mm}
		\textbf{DL delay}:
		Let $S_{d}$ and $S_{d,k,i}$ be the size of the global model update and the size of the split data of the update intended for a UE $k$ in a session $i$, respectively.
		Then, we have
		\begin{align}
			\label{samesizedatad}
			& \sum_{i\in\K} S_{d,k,i} = S_d, \forall k.
		\end{align}
		Let $t_{d,i}$ be the length (in second)
		of the session $i$. Then from Fig.~\ref{fig:time1}(c), the transmission time $t_{d,k,i}$ to the UE $k\in\K$ in the session $i$ of the session-based design is  given by
		\begin{align}\label{tdki}
			t_{d,k,i}(a_{k,i},t_{d,i}) \!=\! a_{k,i} t_{d,i}, \forall k,i.
		\end{align}
		Here, 
		\begin{align}
			\label{sametime:sessisionS1}
			\nonumber
			S_{d,k,i} &= R_{d,k,i}(\ETA_i) t_{d,k,i} \overset{\eqref{tdki}}{=} R_{d,k,i}(\ETA_i) a_{k,i} t_{d,i} \\
			&\overset{\eqref{eta-a-relation} }{=} R_{d,k,i}(\ETA_i) t_{d,i}, \forall k,i.
		\end{align}
		Clearly, \eqref{sametime:sessisionS1} also implies that $(S_{d,k,i} = 0, \text{if $a_{k,i} = 0$}), \forall k,i$, which ensures that no data is sent to the UEs not served in session $i$.
		The transmission time to UE $k\in\K$ in the asynchronous and synchronous designs is expressed as $t_{d,k}(\ETA) \!=\! \frac{S_{d}}{{R_{d,k}(\ETA)}}, \forall k$.
		
		\textbf{Energy consumption for the DL transmission}:
		Denote by $N_0$ the noise power. The energy consumption for transmitting the global update or its split data to a UE $k$ is the product of the transmit power $\rho_d N_0 \eta_{k}$ or $\rho_d N_0 \eta_{k,i}$ and the transmission time to the UE $k$. Therefore, the total energy consumption for transmission by the BS in a session $i$ of the session-based design is $E_{d,i}(\ETA_i, \aaa_i, \ttt_{d}) 
		=\sum_{k\in\K_i} \rho_d N_0 \eta_{k,i}  t_{d,k,i}(a_{k,i},t_{d,i}) = \sum_{k\in\K_i} \rho_d N_0 \eta_{k,i}  a_{k,i} t_{d,i}, \forall k,i,$
		and that in the asynchronous and synchronous designs is $E_{d}(\ETA) 
		=\sum_{k\in\K} \rho_d N_0 \eta_{k}
		\frac{S_{d}}{R_{d,k}(\ETA)}$, 
		where $\aaa_i \triangleq \{a_{k,i}\}_{k\in\K}, \ttt_d \triangleq \{t_{d,i}\}_{i\in\K}$.
		
		\subsubsection{Step (S2)} 
		\vspace{-0mm}
		After receiving the global update, each UE uses its local data set to execute $L$ local computing rounds in order to compute its local update. The model of this step is used in all the proposed designs. 
		
		\textbf{Local computation}:
		Let $c_k$ (cycles/sample) be the number of processing cycles for a UE $k$ to process one data sample \cite{tran19INFOCOM}. Denote by $D_{k}$ (samples) and $f_k$ (cycles/s) the size of the local data set and the processing frequency of the UE $k$, respectively. The computation time at the UE $k$ is then given by $t_{C,k}(f_k) = \frac{LD_kc_k}{f_k}$ \cite{vu20TWC,tran19INFOCOM}.
		
		\textbf{Energy consumption for local computing at the UEs}:
		The energy consumed by the UE $k$ to compute its local training update is given as $E_{C,k}(f_k) = L\frac{\alpha}{2}c_k D_k f_k^2, \forall k$, 
		where $\frac{\alpha}{2}$ is the effective capacitance coefficient of the UEs' computing chipset \cite{tran19INFOCOM,vu20TWC}.
		
		\subsubsection{Step (S3)}
		\vspace{-0mm}
		In this step, the local model updates are transmitted from the UEs to the BS through $K$ sessions. 
		Each coherence block of this step involves two phases: channel estimation and uplink payload data transmission.
		Define the indicator $b_{k,j}$ as 
		\begin{align}\label{b}
			b_{k,j} \triangleq
			\begin{cases}
				1,& \text{if UE $k$ send its data in a session $j$,}\\
				0, & \mbox{otherwise}.
			\end{cases}
		\end{align}
		Let $\NN_j\triangleq\{k|b_{k,j} = 1\}$ be the set of $N_j=\sum_{k\in\K}b_{k,j}
		$ participating UEs in a session $j\in\K$. Here, we have
		\begin{align}
			\label{sumbkj}
			& b_{k,K} = 1, \sum_{k\in\K} b_{k,j} = j, b_{k,j-1} \leq b_{k,j}, \forall k, j
		\end{align}
		to guarantee that all the UEs finish their transmissions in the last session $K$ and each session has one more UE sending its data.
		Doing this helps the UEs that start their transmissions earlier. They can have more power which yields  higher achievable rates, lower delays, and thus, potentially lower transmission energy in each FL communication round. 
		Note that in the asynchronous and synchronous designs, there is only one session $j=K$, and hence, $K_j=K$, $\K_j=\K$, and $\{b_{k,j}\}$ are not variables but constants, i.e., $b_{k,K}=b_k=1, \forall k$. 
		
		\textbf{Uplink channel estimation}:
		In each coherence block, each UE sends its pilot of length $\tau_{u,p}$  to the BS. We assume that the pilots of all the UEs are pairwisely orthogonal, which requires the pilot  lengths to satisfy
		$\tau_{u,p} \geq N_j$.
		The MMSE estimate $\bar{\g}_k$ of $\g_k$ is distributed according to $\CN(\pmb{0},\bar{\sigma}_k^2\pmb{I}_M)$, where
		$\bar{\sigma}_k^2 = \frac{\tau_{u,p} \rho_{p} \beta_k^2}{\tau_{u,p} \rho_{p}\beta_k+1}$ \cite[(3.8)]{ngo16}. 
		
		
		\textbf{UL payload data transmission}:
		After computing the local update, a UE $k$ encodes this update into symbols denoted by $s_{u,k,j}$, where $\EEE\{|s_{u,k,j}|^2\}=1$, and sends the baseband signal $x_{u,k,j}\!=\!\sqrt{\rho_{u}\zeta_{k,j}} s_{u,k,j}$ to the BS, where $\rho_{u}$ is the maximum normalized transmit power at each UE and $\zeta_{k,j}$ is a power control coefficient.
		This signal is subject to the average transmit power constraint,
		$\EEE\left\{|x_{u,k,j}|^2\right\}\leq \rho_u$, which can be expressed as
		\begin{align}\label{poweruupperbound}
			\zeta_{k,j} \leq 1, \forall k\in \NN_j.
		\end{align}
		Here, we have 
		\begin{align}
			\label{zeta-b-relation} 
			(\zeta_{k,j} = 0, \text{if $b_{k,j} = 0$}), \forall k,j
		\end{align}
		to ensure that the UEs not sending  data in session $j$ are not allocated power. 
		After receiving data from all UEs, the BS uses the estimated channels and ZF combining to detect the UEs' message symbols. The ZF receiver requires $M \geq N_j$. The achievable rate
		(bps) of UE $k$ is given by $R_{u,k,j}(\ZETA_j)
		= \frac{\tau_c-\tau_{u,p}}{\tau_c}B \log_2 \big( 1 + \text{SINR}_{u,k,j}(\ZETA_j)\big)$,
		where $\ZETA_j\triangleq\{\zeta_{k,j}\}_{k\in\K}$, $\tau_{u,p}=K$, and
			$\text{SINR}_{u,k,j} (\ZETA_j) \triangleq
			\frac{(M-N_j) \rho_u \bar{\sigma}_{k}^2 \zeta_{k,j}} {\rho_u \sum_{\ell \in \NN_{j}}  (\beta_{\ell} - \bar{\sigma}_{\ell}^2) \zeta_{\ell,j} + 1} \overset{\eqref{zeta-b-relation}}{=} \frac{(M-N_j) \rho_u \bar{\sigma}_{k}^2 \zeta_{k,j}} {\rho_u \sum_{\ell \in \K}  (\beta_{\ell} - \bar{\sigma}_{\ell}^2) \zeta_{\ell,j} + 1}$
		is the effective uplink SINR \cite[(3.29)]{ngo16}.
		
		Similarly, the power constraint at the UEs and the achievable rate of the UE $k$ in the asynchronous and synchronous designs are given by
		\begin{align}
			\label{poweruupperbound:asyn:syn}
			&\zeta_{k} \leq 1, \forall k\in \K
			\\
			\label{Ru}
			&R_{u,k}(\ZETA)
			= \frac{\tau_c-\tau_{u,p}}{\tau_c}B \log_2 \big( 1 + \text{SINR}_{u,k}(\ZETA)\big),
		\end{align}
		where $\ZETA\triangleq\{\zeta_{k}\}_{k\in\K}$ are power control coefficients, and $\text{SINR}_{u,k} (\ZETA) \triangleq
		\frac{(M-K) \rho_u \bar{\sigma}_{k}^2 \zeta_{k}} {\rho_u \sum_{\ell \in \K}  (\beta_{\ell} - \bar{\sigma}_{\ell}^2) \zeta_{\ell} + 1}$.
		
		\textbf{UL delay}:
		Let $S_{u}$ and $S_{u,k,j}$ be the size of the local model update and the size of the split data of this update in a session $j$, respectively. Then, we have
		\begin{align}
			\label{samesizedatau}
			&\sum_{j\in\K} S_{u,k,j} = S_u, \forall k.
		\end{align}
		Since the transmission time $t_{u,j}$ from every participating UE to the BS in the session $j$ is the same, the transmission time $t_{u,k,j}$ from a UE $k\in\K$ in the session-based design is given by
		\begin{align}\label{tukj}
			\!\!\! t_{u,k,j}(b_{k,j},t_{u,j}) = b_{k,j} t_{u,j}, \forall k,j,
		\end{align}
		and 
		\begin{align}
			\label{sametime:sessisionS3}
			\nonumber
			S_{u,k,j} &= R_{u,k,j}(\ZETA_j) t_{u,k,j} \overset{\eqref{tukj}}{=} R_{u,k,j}(\ZETA_j) b_{k,j}  t_{u,j} 
			\\
			&\overset{\eqref{zeta-b-relation}}{=} R_{u,k,j}(\ZETA_j) t_{u,j}, \forall k,j.
		\end{align}
		Here, \eqref{sametime:sessisionS3} also implies $(S_{u,k,j} = 0, \text{if $b_{k,j} = 0$}), \forall k,j$, which ensures that the UEs not participating in session the $j$ do not send any data. 
		The transmission time from the UE $k\in\K$ in the asynchronous and synchronous designs is $t_{u,k}(\ZETA) =  \frac{S_{u}}{R_{u,k}(\ZETA)}, \forall k,j$.
		
		\textbf{Energy consumption for the UL transmission}:
		The energy consumption for the UL transmission at a UE is the product of the UL power and the transmission time. In particular, the energy consumption at a UE $k$ in a session $j$ of the session-based design is given by $E_{u,k,j}(\zeta_{k,j},b_{k,j},t_{u,j}) = \rho_u N_0 \zeta_{k,j} t_{u,k,j}(b_{k,j},t_{u,j}) = \rho_u N_0 \zeta_{k,j} b_{k,j} t_{u,j}, \forall k,j$,
		and that of both the asynchronous and synchronous designs is expressed as $E_{u,k}(\ZETA)= \frac{\rho_u N_0\zeta_{k} S_{u}} {R_{u,k}(\ZETA)}, \forall k$.
		
		\subsubsection{Step (S4)}
		In this step, the BS recomputes the global update using all the received local updates. 
		This step is executed at the BS and does not affect our transmission designs. The computational capability of the central server (i.e., the BS) is much higher than that of each UE, and Step (S4) typically entails the application of a simple aggregation rule such as summing up the model updates. Therefore, the time required for computing the global update in Step (S4) is assumed negligible. 
		Consequently, the computation time of Step (S4) is ignored in the problem formulation and solution in the subsequent sections.

		\vspace{-0mm}
		\section{Session-based Scheme: Problem Formulation and Solution}
		\label{sec:PF}
		\vspace{-0mm}
		\subsection{Problem Formulation}
		\vspace{-0mm}
		In this work, we aim to (i) improve the energy efficiency of the proposed FL-enabled mMIMO networks by minimizing the total energy consumption in one FL communication round, and (ii) guarantee the execution time
		of each round below a quality-of-service threshold. Here, the total energy consumption of one FL communication round includes the energy consumption for transmission and local computation at both the BS and the UEs. Thus, the total energy consumption of one FL communication round in the session-based design is
			$E_{SB}(\aaa,\bb,\widehat{\ETA}, \widehat{\ZETA}, \f, \ttt_d, \ttt_u) \triangleq 
			\sum_{i\in\K} E_{d,i}(\ETA_i,\aaa_i,\ttt_{d}) + \sum_{k\in\K} \big( E_{C,k}(f_k) + \sum_{j\in\K} E_{u,k,j}(\zeta_{k,j},b_{k,j},t_{u,j}) \big)
			= 
			\sum_{k\in\K}\sum_{i\in\K} \rho_d N_0 \eta_{k,i} a_{k,i} t_{d,i}
			+ \sum_{k\in\K} \!\big( L\frac{\alpha}{2} c_{k}D_kf_{k}^2 + \!\sum_{j\in\K}\! \rho_u N_0\zeta_{k,j} b_{k,j} t_{u,j} \big),$
		with $\aaa\! \triangleq \{a_{k,i}\}, \bb\! \triangleq \{b_{k,j}\}, \widehat{\ETA} \triangleq \{\ETA_{i}\}, \widehat{\ZETA} \triangleq \{\ZETA_{j}\}, \f \triangleq \{f_k\}, \ttt_u \triangleq \{t_{u,j}\}, \forall k,i,j$.

		\vspace{-0mm}
		The problem of optimizing user assignment $(\aaa,\bb)$, data size $(\Ss_d,\Ss_u)$, time allocation $(\ttt_d, \ttt_u)$, power $(\widehat{\ETA},\widehat{\ZETA})$, and computing frequency $\f$, to minimize the total energy consumption of one FL communication round in the session-based design is formulated as
		\begin{subequations}\label{Pmain:sb}
			\begin{align}
				\label{CFsb}
				\!\!\!\!\!\underset{\x}{\min} \,\,
				& E_{SB} (\aaa,\bb,\widehat{\ETA},\widehat{\ZETA}, \f, \ttt_d, \ttt_u)
				\\
				\!\!\!\!\!\mathrm{s.t.}\,\,
				\nonumber
				& \eqref{a}-
				\eqref{eta-a-relation},  
				\eqref{samesizedatad}, 
				\eqref{sametime:sessisionS1}, 
				\eqref{b}-
				\eqref{zeta-b-relation}, \eqref{samesizedatau}, 
				\eqref{sametime:sessisionS3}
				\\
				\label{powerlowerbound}
				& 0\leq \eta_{k,i}, 0\leq \zeta_{k,j}, \forall k,i,j
				\\
				\label{fbound}
				& 0 \leq f_{k} \leq f_{\max}, \forall k
				\\
				\label{QoSbound-1}
				& \sum_{i\in\K}t_{d,k,i}
				\!+\! t_{C,k}
				+ \sum_{j\in\K}t_{u,k,j}
				= t, \forall k
				\\
				\label{QoSbound-2}
				&t \leq t_{\text{QoS}}
				\\
				\label{syncbound}
				& \max_{k\in\K} \sum_{i\in\K}t_{d,k,i}
				\!\leq\! \min_{k\in\K} \big(\sum_{i\in\K}t_{d,k,i}
				\!+\! t_{C,k} 
				\big),
			\end{align}
		\end{subequations}
		where $\x \triangleq \{\aaa, \bb, \widehat{\ETA}, \widehat{\ZETA}, \f, \Ss_d, \Ss_u, \ttt_d, \ttt_u\}, \Ss_d\! \triangleq \{S_{d,k,i}\}, \Ss_u\! \triangleq \{S_{u,k,j}\}, \forall k,i,j$.
		Here, \eqref{syncbound} is introduced to ensure that all the UEs send their local updates during the UL mode of the BS. The right-hand side of \eqref{syncbound} corresponds to the first UE that finishes its DL transmission and local computation, while the left-hand side corresponds to the slowest UE that finishes its DL transmission. Constraints \eqref{QoSbound-1} and \eqref{QoSbound-2} take into account the time consumption in each FL communication round. These constraints make sure that the time consumption of each FL communication round does not exceed the minimum requirement $t_{\text{QoS}}$, in order to ensure a target level of quality of service. Note that the study of the optimal trade-off between the time and energy consumption such as \cite{luo20TWC} is interesting but beyond the scope of our paper, and hence, is left for future work.
		
		\vspace{-0mm}
		\begin{remark}
			\label{remark:nolearning}
			Similar to many existing works that follow the DigComp approach   (such as \cite{yang21TWC,zeng21TWC,wu22TCOM,viet22TVT,pavlos22CL}), our intention is to design energy-efficient wireless networks to support standard FL. 
			Also, we do not combine massive MIMO and FL to create a new learning framework. 
			We focus on the communication aspects, and more specifically the  schemes for users to receive, compute, and transmit their model updates. On one hand, our proposed schemes do not require any changes to or even assumptions on  the learning algorithm. As such,  the learning performance (including convergence rates) of any standard FL framework (e.g., those in \cite{canh21TN,sai20ICML,canh20NIPS,felix20TNNLS,tran19INFOCOM,mcmahan17AISTATS}) implemented over massive MIMO systems using our proposed schemes remains unchanged. The complexity of the existing FL algorithm to be implemented on the proposed massive MIMO networks does not increase. On the other hand, transmitting and receiving FL model updates is nothing but transmitting and receiving data between user devices and base station. Therefore, the complexity of a massive MIMO network used to support FL is similar to that of a current 5G massive MIMO network with the same system configuration. 
		\end{remark}

		
		\vspace{-0mm}
		\subsection{Solution}
		\vspace{-0mm}
		\label{sec:alg:sb}
		Finding a globally optimal solution to problem \eqref{Pmain:sb} is challenging due to the mixed-integer and nonconvex constraints \eqref{a}, \eqref{eta-a-relation},  \eqref{sametime:sessisionS1}, \eqref{b}, \eqref{zeta-b-relation}, \eqref{sametime:sessisionS3}, \eqref{QoSbound-1}, and \eqref{syncbound}. Therefore, we instead propose a solution approach that is suitable for practical implementation. 
		
		\subsubsection{Problem Transformation}
		First, we aim to transform the problem \eqref{Pmain:sb} into a more tractable form. Specifically, 
		we replace constraints \eqref{eta-a-relation} and \eqref{zeta-b-relation} by
		\begin{align}
			\label{power-a-b-relation}
			\eta_{k,i} \leq a_{k,i}, \zeta_{k,i}\leq b_{k,j},\forall k,i,j.
		\end{align}
		Then, constraints \eqref{sametime:sessisionS1} and \eqref{sametime:sessisionS3} are replaced by
		\begin{align}
			\label{sametime:sessisionS1-1}
			& S_{d,k,i} \leq R_{d,k,i}(\ETA_i) a_{k,i} t_{d,i}, \forall k,i
			\\
			\label{sametime:sessisionS1-2}
			& S_{d,k,i} \geq R_{d,k,i}(\ETA_i) a_{k,i} t_{d,i}, \forall k,i
			\\
			\label{sametime:sessisionS3-1}
			& S_{u,k,j} \leq R_{u,k,j}(\ZETA_j) b_{k,j} t_{u,j}, \forall k,j
			\\
			\label{sametime:sessisionS3-2}
			& S_{u,k,j} \geq R_{u,k,j}(\ZETA_j) b_{k,j} t_{u,j}, \forall k,j.
		\end{align}
		Constraints \eqref{sametime:sessisionS1-1}--\eqref{sametime:sessisionS3-2} are further replaced by
		\begin{align}
			\label{ratedki-lowerbound-1}
			& \hat{r}_{d,k,i} \leq R_{d,k,i}(\ETA_i), \forall k,i
			\\
			\label{ratedki-upperbound-1}
			& \tilde{r}_{d,k,i} \geq R_{d,k,i}(\ETA_i), \forall k,i
			\\
			\label{tdki-lowerbound}
			& \hat{t}_{d,k,i} \leq a_{k,i} t_{d,i}, \forall k,i
			\\
			\label{tdki-upperbound}
			& \tilde{t}_{d,k,i} \geq a_{k,i} t_{d,i}, \forall k,i
			\\
			\label{sametime:sessisionS1-1-a}
			& S_{d,k,i} \leq \hat{r}_{d,k,i} \hat{t}_{d,k,i}, \forall k,i
			\\
			\label{sametime:sessisionS1-2-a}
			& S_{d,k,i} \geq \tilde{r}_{d,k,i} \tilde{t}_{d,k,i}, \forall k,i
			\\
			\label{rateukj-lowerbound-1}
			& \hat{r}_{u,k,j} \leq R_{u,k,j}(\ZETA_j), \forall k,j
			\\
			\label{rateukj-upperbound-1}
			& \tilde{r}_{u,k,j} \geq R_{u,k,j}(\ZETA_j), \forall k,j
			\\
			\label{tukj-lowerbound}
			& \hat{t}_{u,k,j} \leq b_{k,j} t_{u,j}, \forall k,j
			\\
			\label{tukj-upperbound}
			& \tilde{t}_{u,k,j} \geq b_{k,j} t_{u,i}, \forall k,j
			\\
			\label{sametime:sessisionS3-1-a}
			& S_{u,k,j} \leq \hat{r}_{u,k,j} \hat{t}_{u,k,j}, \forall k,j
			\\
			\label{sametime:sessisionS3-2-a}
			& S_{u,k,j} \geq \tilde{r}_{u,k,j} \tilde{t}_{u,k,j}, \forall k,j,
		\end{align}
		where $\hat{\rr}_d \triangleq \{\hat{r}_{d,k,i}\}, \tilde{\rr}_d \triangleq \{\tilde{r}_{d,k,i}\}, \hat{\ttt}_d \triangleq \{\hat{t}_{d,k,i}\}, \tilde{\ttt}_d \triangleq \{\tilde{t}_{d,k,i}\}, \hat{\rr}_u \triangleq \{\hat{r}_{u,k,j}\}, \tilde{\rr}_u \triangleq \{\tilde{r}_{u,k,j}\}, \hat{\ttt}_u \triangleq \{\hat{t}_{u,k,j}\}, \tilde{\ttt}_u \triangleq \{\tilde{t}_{u,k,j}\}$ are additional variables. 
		By using \eqref{tdki-upperbound}, \eqref{tukj-lowerbound} and \eqref{tukj-upperbound}, we replace constraint \eqref{QoSbound-1} by 
		\begin{align}
			\label{QoSbound-1-a}
			&\sum_{i\in\K}\tilde{t}_{d,k,i} \!+\! t_{C,k}(f_k) \!+\! \sum_{i\in\K}\tilde{t}_{u,k,i}
			\leq t, \forall k
			\\
			\label{QoSbound-1-b}
			&\sum_{i\in\K}\hat{t}_{d,k,i} \!+\! t_{C,k}(f_k) \!+\! \sum_{i\in\K}\hat{t}_{u,k,i}
			\geq t, \forall k,
		\end{align}
		and constraint \eqref{syncbound} by
		\begin{align}
			\label{syncbound-1}
			& \sum_{i\in\K} \tilde{t}_{d,k,i} \leq q, \forall k
			\\
			\label{syncbound-2}
			& q \leq q_{1,k} + q_{2,k}, \forall k
			\\
			\label{syncbound-3}
			& q_{1,k} \leq \sum_{i\in\K} \hat{t}_{d,k,i}, \forall k
			\\
			\label{syncbound-4}
			& q_{2,k} \leq \frac{LD_kc_k}{f_k}, \forall k,
		\end{align}
		where $t, q, \pmb{q}_1 \triangleq \{q_{1,k}\}, \pmb{q}_2 \triangleq \{q_{2,k}\}$ are additional variables.  
		Now, problem \eqref{Pmain:sb} can be transformed into its more tractable epigraph form as
		\begin{subequations}\label{Pmain:epi}
			\begin{align}
				\label{CFPmulti}
				\underset{\widetilde{\x}}{\min} \,\,
				& \widetilde{E}_{SB}(\f,\vv_d,\vv_u) 
				\\
				\mathrm{s.t.}\,\,
				\nonumber
				& \eqref{a}-
				\eqref{powerdupperbound},
				\eqref{samesizedatad},
				\eqref{b}-
				\eqref{poweruupperbound}, 
				\eqref{samesizedatau}, 
				\eqref{powerlowerbound}, \eqref{fbound}, \eqref{ratedki-lowerbound-1}-\eqref{syncbound-4}
				\\
				\label{obj-upperbound-1}
				& \eta_{k,i} \tilde{t}_{d,k,i} \leq v_{d,k,i}, \forall k,i
				\\
				\label{obj-upperbound-2}
				& \zeta_{k,j} \tilde{t}_{u,k,j} \leq v_{u,k,j}, \forall k,j,
			\end{align}
		\end{subequations}
		where $\widetilde{\x} \triangleq \{\x, \hat{\rr}_d, \tilde{\rr}_d, \hat{\ttt}_d, \tilde{\ttt}_d, \hat{\rr}_u, \tilde{\rr}_u, \hat{\ttt}_u, \tilde{\ttt}_u, t, q, \pmb{q}_1, \pmb{q}_2, \vv_d, \vv_u\}$, 
		$\vv_d \triangleq \{v_{d,k,i}\}, \vv_u \triangleq \{v_{u,k,j}\}$ are additional variables, and $\widetilde{E}_{SB}(\f,\vv_d,\vv_u) \triangleq \sum_{k\in\K}\sum_{i\in\K} \rho_d N_0 v_{d,k,i} +
		\sum_{k\in\K} \!\Big( L\frac{\alpha}{2}c_{k}D_kf_{k}^2 \!+\!\sum_{j\in\K} \rho_u N_0 v_{u,k,j} \Big)$.
		
		Next, to deal with the binary constraints \eqref{a} and \eqref{b}, we observe that $x\in\{0,1\}\Leftrightarrow x\in[0,1]\,\&\,x-x^2\leq0$ \cite{vu18TCOM,vu18TWC}. Therefore, \eqref{a} and \eqref{b} are equivalent to the following constraints
		\begin{align}
			\label{sumab}
			\!\!\!\!\!& V_1
			\triangleq\!\! \sum_{k\in\K} \sum_{i\in\K} (a_{k,i}\!-\!a_{k,i}^2) \!+\!\!  \sum_{k\in\K} \sum_{j\in\K} (b_{k,j}\!-\!b_{k,j}^2) \leq 0
			\\
			\label{abrelax}
			& 0\leq a_{k,i}\! \leq 1, 0\leq b_{k,j}\!\leq 1, \forall k,i,j.
		\end{align}
		Also, we observe from \eqref{ratedki-lowerbound-1}--\eqref{sametime:sessisionS1-1-a} and \eqref{rateukj-lowerbound-1}--\eqref{sametime:sessisionS3-1-a} that $S_{d,k,i} \leq \tilde{r}_{d,k,i} \tilde{t}_{d,k,i}, S_{u,k,j} \leq \tilde{r}_{u,k,j} \tilde{t}_{u,k,j} \forall k,i,j$. Thus, \eqref{sametime:sessisionS1-2-a} and \eqref{sametime:sessisionS3-2-a} are equivalent to the following constraints
		\begin{align}
			\label{V2}
			V_2
			&\triangleq \sum_{k\in\K}\sum_{i\in\K} (\tilde{r}_{d,k,i} \tilde{t}_{d,k,i} - S_{d,k,i}) \leq 0
			\\
			\label{V3}
			V_3
			&\triangleq \sum_{k\in\K}\sum_{i\in\K} (\tilde{r}_{u,k,j} \tilde{t}_{u,k,j} - S_{u,k,j}) \leq 0.
		\end{align}
		Similarly, from \eqref{QoSbound-1-a}, constraint \eqref{QoSbound-1-b} is equivalent to
		\begin{align}
			\label{V4}
			V_4
			\triangleq \sum_{k\in\K} \big(t - \sum_{i\in\K}\hat{t}_{d,k,i} \!-\! t_{C,k}(f_k) \!-\! \sum_{i\in\K}\hat{t}_{u,k,i}\big)  \leq 0.
		\end{align}
		Therefore, problem \eqref{Pmain:epi} is equivalent to
		\begin{align}
			\label{Pmain:epi:equiv}
			\underset{\widetilde{\x}\in\FF}{\min} \,\,
			& \widetilde{E}_{SB}(\f,\vv_d,\vv_u),
		\end{align}
		where $\FF\! \triangleq\! \{\eqref{sumaki}, \eqref{powerdupperbound}, \eqref{samesizedatad},
		\eqref{sumbkj}, \eqref{poweruupperbound}, \eqref{samesizedatau},
		\eqref{powerlowerbound}, \eqref{fbound}, \eqref{ratedki-lowerbound-1}-\eqref{sametime:sessisionS1-1-a}, \eqref{rateukj-lowerbound-1}-\eqref{sametime:sessisionS3-1-a}, \eqref{QoSbound-1}-\eqref{syncbound-4}, \eqref{obj-upperbound-1},\\ \eqref{obj-upperbound-2},\eqref{sumab}-\eqref{V4}\}$. 
		Then, we consider the problem
		\begin{align}\label{Pmain:epi:relax}
			\underset{\widetilde{\x} \in \widehat{\FF}}{\min} \,\,
			&\LL (\widetilde{\x})
			,
		\end{align}
		where $\LL (\widetilde{\x})
		\triangleq \widetilde{E}_{SB}(\f,\vv_d,\vv_u) \!+\! \lambda( \gamma_1V_1(\aaa,\bb) + \gamma_2 V_2(\tilde{\rr}_d, \tilde{\ttt}_d, \Ss_d) + \gamma_3 V_3(\tilde{\rr}_u, \tilde{\ttt}_u, \Ss_u) + \gamma_4 V_4(\f, \hat{\ttt}_d, \hat{\ttt}_u))$ is the Lagrangian of \eqref{Pmain:epi:equiv}, $\gamma_1, \gamma_2, \gamma_3, \gamma_4 > 0$ are fixed weights, and $\lambda$ is the Lagrangian multiplier corresponding to constraints \eqref{sumab}--\eqref{V4}. Here, $\widehat{\FF} \triangleq \FF \setminus \{\eqref{sumab}-\eqref{V4}\}$.
		\vspace{-0mm}
		\begin{proposition}
			\label{proposition-dual}
			The following statements hold:
			\renewcommand{\labelenumi}{(\roman{enumi})}
			\begin{enumerate}
				\item The values $V_{1,\lambda}$, $V_{2,\lambda}$, $V_{3,\lambda}$, $V_{4,\lambda}$ of $V_1, V_2, V_3, V_4$ at the solution of \eqref{Pmain:epi:relax} corresponding to $\lambda$ converge to $0$ as $\lambda \rightarrow +\infty$.
				\item Problem \eqref{Pmain:epi:equiv} has the following property
				\begin{equation}\label{Strong:Dualitly:hold}
					\underset{\widetilde{\x}\in\FF}{\min}\,\,
					\widetilde{E}_{SB}(\f,\vv_d,\vv_u)
					=
					\underset{\lambda\geq0}{\sup}\,\,
					\underset{\widetilde{\x}\in\widehat{\FF}}{\min}\,\,
					\LL (\widetilde{\x})
					,
				\end{equation}
				and therefore, it is equivalent to \eqref{Pmain:epi:relax}  at the optimal solution $\lambda^* \geq0$ of the sup-min problem in \eqref{Strong:Dualitly:hold}.
			\end{enumerate}
		\end{proposition}
		\begin{proof}
			See the Appendix.
			\vspace{-0mm}
		\end{proof}
		Proposition~\ref{proposition-dual} suggests that we can obtain the optimal solution to problem \eqref{Pmain:epi:equiv} by solving \eqref{Pmain:epi:relax} with appropriately chosen parameters $\lambda$, $\gamma_1$, $\gamma_2$, $\gamma_3$, and $\gamma_4$.
		Theoretically, it is required to have $V_{1,\lambda}=0$, $V_{2,\lambda}=0$, $V_{3,\lambda}=0$, and $V_{4,\lambda}=0$ in order to obtain the optimal solution to problem \eqref{Pmain:epi:equiv}. According to Proposition~\ref{proposition-dual}, $V_{1,\lambda}$, $V_{2,\lambda}$, $V_{3,\lambda}$, and $V_{4,\lambda}$ converge to $0$ as $\lambda\to+\infty$. Since there is always a numerical tolerance
		in computation, it is sufficient to accept $V_{1,\lambda}\leq\varepsilon$, $V_{2,\lambda}\leq\varepsilon$, $V_{3,\lambda}\leq\varepsilon$, and $V_{4,\lambda}\leq\varepsilon$ for some small $\varepsilon$ with a sufficiently large value of $\lambda$.
		In our numerical experiments, for $\varepsilon = 10^{-3}$, we see that choosing $\lambda=1$ with $\gamma_1=0.1, \gamma_2=0.01, \gamma_3=0.01, \gamma_4=0.01$ is enough to ensure $V_{1,\lambda}\leq\varepsilon$, $V_{2,\lambda}\leq\varepsilon$, $V_{3,\lambda}\leq\varepsilon$, and $V_{4,\lambda}\leq\varepsilon$. This way of choosing $\lambda$ has been widely used in the literature, e.g., \cite{vu18TWC,vu18TCOM,vu21ICC,che14TWC,Rashid14TCOM}.
		
		\vspace{-0mm}
		\subsubsection{Algorithm Development}
		Problem \eqref{Pmain:epi:relax} is still difficult to solve due to nonconvex constraints \eqref{ratedki-lowerbound-1}--\eqref{sametime:sessisionS1-1-a}, \eqref{rateukj-lowerbound-1}--\eqref{sametime:sessisionS3-1-a}, \eqref{syncbound-4}, \eqref{obj-upperbound-1}, \eqref{obj-upperbound-2} and nonconvex parts $V_1(\aaa,\bb), V_2(\tilde{\rr}_d, \tilde{\ttt}_d, \Ss_d), V_3(\tilde{\rr}_u, \tilde{\ttt}_u, \Ss_u), V_4(\f, \hat{\ttt}_d, \hat{\ttt}_u)$ in the cost function $\LL (\widetilde{\x})
		$.  To deal with constraints \eqref{ratedki-lowerbound-1} and  \eqref{rateukj-lowerbound-1}, we observe that 
			$\log\left(1+\frac{x}{y}\right) \geq \log\left(1+\frac{x^{(n)}}{y^{(n)}}\right) + \frac{2x^{(n)}}{(x^{(n)}+y^{(n)})} 
			- \frac{(x^{(n)})^2}{(x^{(n)}+y^{(n)})x}
			- \frac{x^{(n)}y}{(x^{(n)}+y^{(n)})y^{(n)}},$
		for any $x > 0, y > 0, x^{(n)}>0$, and $y^{(n)}>0$ \cite[(76)]{long21TCOM}.
		Therefore, the concave lower bounds 
		$\widehat{R}_{d,k,i} (\ETA_i)
		$ and 
		$\widehat{R}_{u,k,j}(\ZETA_j)
		$ of $R_{d,k,i}(\ETA_i)$ and $R_{u,k,j}(\ZETA_j)$ are, respectively, given by
			$\widehat{R}_{d,k,i} \triangleq \frac{\tau_c - \tau_{d,p}}{\tau_c\log 2} B \Big[ \log\big(1+\frac{\Upsilon_i^{(n)}}{\Phi_i^{(n)}}\big) + \frac{2\Upsilon_i^{(n)}}{(\Upsilon_i^{(n)}+\Phi_i^{(n)})} 
			- \frac{(\Upsilon_i^{(n)})^2}{(\Upsilon_i^{(n)}+\Phi_i^{(n)})\Upsilon_i}
			- \frac{\Upsilon_i^{(n)}\Phi}{(\Upsilon_i^{(n)}+\Phi_i^{(n)})\Phi_i^{(n)}}\Big],
			\widehat{R}_{u,k,j} \triangleq \frac{\tau_c - \tau_{u,p}}{\tau_c\log 2} B \Big[ \log\big(1+\frac{\Psi_i^{(n)}}{\Omega_i^{(n)}}\big) + \frac{2\Psi_i^{(n)}}{(\Psi_i^{(n)}+\Omega_i^{(n)})} 
			- \frac{(\Psi_i^{(n)})^2}{(\Psi_i^{(n)}+\Omega_i^{(n)})\Psi_i}
			- \frac{\Psi_i^{(n)}\Omega_i}{(\Psi_i^{(n)}+\Omega_i^{(n)})\Omega_i^{(n)}}\Big]$, 
		where $\Upsilon_i(\eta_{k,i}) \triangleq (M-K_i)\rho_d \hat{\sigma}_k^2\eta_{k,i}$, $\Phi_i(\ETA_i) \triangleq \rho_d (\beta_k - \hat{\sigma}_k^2) \sum_{\ell\in\K} \eta_{\ell,i} +1, \Psi_i(\zeta_{k,j}) \triangleq (M-N_j) \rho_u \bar{\sigma}_{k}^2 \zeta_{k,j}$, and $\Omega_i(\ZETA_j) \triangleq \rho_u \sum_{\ell \in \NN_{j}}  (\beta_{\ell} - \bar{\sigma}_{\ell}^2) \zeta_{\ell,j} + 1$. 
		Then, constraints \eqref{ratedki-lowerbound-1} and  \eqref{rateukj-lowerbound-1} can be approximated by the following convex constraints
		\begin{align}
			\label{ratedki-lowerbound-1-a}
			& \hat{r}_{d,k,i} \leq \widehat{R}_{d,k,i}(\ETA_i), \forall k,i
			\\
			\label{rateukj-lowerbound-1-a}
			& \hat{r}_{u,k,j} \leq \hat{R}_{u,k,j}(\ZETA_j), \forall k,i.
		\end{align}
		
		To deal with the constraints \eqref{ratedki-upperbound-1} and  \eqref{rateukj-upperbound-1}, we observe that
			$\log\big(1+\frac{x}{y}\big) \leq H(x,y)\triangleq 
			\log\big(x^{(n)}+y^{(n)}\big) + \frac{x + y - x^{(n)} - y^{(n)}}{x^{(n)}+y^{(n)}} - \log(y),$
		where $x > 0,y > 0, x^{(n)}>0$, and $y^{(n)}>0$.
		Therefore, the convex upper bounds $\widetilde{R}_{d,k,i}(\ETA_i)$ and $\widetilde{R}_{u,k,j}(\ZETA_j)$ of 
		$R_{d,k,i}(\ETA_i)$ and $R_{u,k,j}(\ZETA_j)$ can be respectively expressed as
			$\widetilde{R}_{d,k,i} \triangleq \frac{\tau_c - \tau_{u,p}}{\tau_c\log 2} B \Big[\log\big(\Upsilon_i^{(n)}+\Phi_i^{(n)}\big) + \frac{\Upsilon_i + \Phi_i - \Upsilon_i^{(n)} - \Phi_i^{(n)}}{\Upsilon_i^{(n)}+\Phi_i^{(n)}} - \log(\Phi_i)\Big],
			\widetilde{R}_{u,k,j} \triangleq \frac{\tau_c - \tau_{u,p}}{\tau_c\log 2} B \Big[\log\big(\Psi_i^{(n)}+\Omega_i^{(n)}\big) + \frac{\Psi_i + \Omega_i - \Psi_i^{(n)} - \Omega_i^{(n)}}{\Psi_i^{(n)}+\Omega_i^{(n)}} - \log(\Omega_i) \Big].$
		The constraints \eqref{ratedki-upperbound-1} and  \eqref{rateukj-upperbound-1} can be approximated by the following convex constraints
		\begin{align}
			\label{ratedki-upperbound-1-a}
			& \tilde{r}_{d,k,i} \geq \widetilde{R}_{d,k,i}(\ETA_i) , \forall k,i
			\\
			\label{rateukj-upperbound-1-a}
			& \tilde{r}_{u,k,j} \geq \widetilde{R}_{u,k,j}(\ZETA_j), \forall k,i.
		\end{align}
		
		Next, to deal with the constraints \eqref{tdki-lowerbound}--\eqref{sametime:sessisionS1-1-a}, \eqref{tukj-lowerbound}--\eqref{sametime:sessisionS3-1-a}, \eqref{obj-upperbound-1}, and \eqref{obj-upperbound-2},
		we observe that 
		$xy-z \leq 0.25 [(x+y)^2-2(x^{(n)}-y^{(n)})(x-y) + (x^{(n)}-y^{(n)})^2 - 4z]$ and $z-xy \leq 0.25 [4z\! +\! (x\!-\!y)^2\!-\!2(x^{(n)}\!+\!y^{(n)})(x\!+\!y)
		+ (x^{(n)}+y^{(n)})^2], \forall x, y, z, x^{(n)}, y^{(n)}, z^{(n)}$ \cite{vu20TWC}. 
		Therefore, the constraints \eqref{tdki-lowerbound}--\eqref{sametime:sessisionS1-1-a}, \eqref{tukj-lowerbound}--\eqref{sametime:sessisionS3-1-a}, \eqref{obj-upperbound-1}, and \eqref{obj-upperbound-2} can be approximated respectively by the following convex constraints
		\begin{align}
			\nonumber
			\label{tdki-lowerbound-1}
			& 0.25 [4\hat{t}_{d,k,i} + (a_{k,i}-t_{d,i})^2 -2(a_{k,i}^{(n)}+t_{d,i}^{(n)})(a_{k,i}+t_{d,i})
			\\
			& \qquad + (a_{k,i}^{(n)}+t_{d,i}^{(n)})^2] \leq 0, \forall k,i
			\\
			\nonumber
			&  0.25[(a_{k,i}+t_{d,i})^2-2(a_{k,i}^{(n)}-t_{d,i}^{(n)})(a_{k,i}-t_{d,i})
			\\
			& \qquad + (a_{k,i}^{(n)}-t_{d,i}^{(n)})^2 - 4\tilde{t}_{d,k,i}] \leq 0, \forall k,i
			\\
			\nonumber
			& 0.25 [4S_{d,k,i}\! +\! (\hat{r}_{d,k,i}\!-\!\hat{t}_{d,k,i})^2\!-\!2(\hat{r}_{d,k,i}^{(n)}\!+\!\hat{t}_{d,k,i}^{(n)})(\hat{r}_{d,k,i}\!+\!\hat{t}_{d,k,i})
			\\
			& \qquad + (\hat{r}_{d,k,i}^{(n)}+\hat{t}_{d,k,i}^{(n)})^2] \leq 0, \forall k,i
			\\
			\nonumber
			& 0.25 [4\hat{t}_{u,k,j}\! +\! (b_{k,j}\!-\!t_{u,j})^2\!-\!2(b_{k,j}^{(n)}\!+\!t_{u,j}^{(n)})(b_{k,j}\!+\!t_{u,j})
			\\
			& \qquad+ (b_{k,j}^{(n)}+t_{u,j}^{(n)})^2] \leq 0, \forall k,j
			\\
			\nonumber
			& 0.25 [(b_{k,j}+t_{u,j})^2-2(b_{k,j}^{(n)}-t_{u,j}^{(n)})(b_{k,j}-t_{u,j})
			\\
			& \qquad+ (b_{k,j}^{(n)}-t_{u,j}^{(n)})^2 - 4\tilde{t}_{u,k,j}] \leq 0, \forall k,j
			\\
			\nonumber
			& 0.25 [4S_{u,k,j}\! +\! (\hat{r}_{u,k,j}\!-\!\hat{t}_{u,j})^2\!-\!2(\hat{r}_{u,k,j}^{(n)}\!+\!\hat{t}_{u,j}^{(n)})(\hat{r}_{u,k,j}\!+\!\hat{t}_{u,j})
			\\
			& \qquad+ (\hat{r}_{u,k,j}^{(n)}+\hat{t}_{u,j}^{(n)})^2] \leq 0, \forall k,j
			\\
			\label{obj-upperbound-1-a}
			\nonumber
			& 0.25 [(\eta_{k,i}+\tilde{t}_{d,k,i})^2-2(\eta_{k,i}^{(n)}-\tilde{t}_{d,k,i}^{(n)})(\eta_{k,i}-\tilde{t}_{d,k,i})
			\\
			& \qquad+ (\eta_{k,i}^{(n)}-\tilde{t}_{d,k,i}^{(n)})^2 - 4v_{d,k,i}] \leq 0, \forall k,i
			\\
			\nonumber
			\label{obj-upperbound-2-a}
			& 0.25 [(\zeta_{k,j}+\tilde{t}_{u,k,j})^2-2(\zeta_{k,j}^{(n)}-\tilde{t}_{u,k,j}^{(n)})(\zeta_{k,j}-\tilde{t}_{u,k,j})
			\\
			& \qquad+ (\zeta_{k,j}^{(n)}-\tilde{t}_{u,k,j}^{(n)})^2 - 4v_{u,k,j}] \leq 0, \forall k,j.
		\end{align}
		Similarly, 
		the convex upper bounds 
		of the nonconvex parts $V_1(\aaa,\bb), V_2(\tilde{\rr}_d, \tilde{\ttt}_d, \Ss_d), V_3(\tilde{\rr}_u, \tilde{\ttt}_u, \Ss_u)$ are respectively given by
			$\widetilde{V}_1
			\!\triangleq\! \sum_{i\in\K}\sum_{k\in\NN} (a_{k,i}-2a_{k,i}^{(n)}a_{k,i} + (a_{k,i}^{(n)})^2) \!+\!\! \sum_{j\in\K}\sum_{k\in\NN} (b_{k,j}-2b_{k,j}^{(n)}b_{k,j} + (b_{k,j}^{(n)})^2),
			\widetilde{V}_2
			\! \triangleq\!\! \sum_{i\in\K}\sum_{k\in\NN}\!
			0.25 [(\tilde{r}_{d,k,i}\!+\!\tilde{t}_{d,k,i})^2\!-\!2(\tilde{r}_{d,k,i}^{(n)}\!-\!\tilde{t}_{d,k,i}^{(n)})(\tilde{r}_{d,k,i}\!-\!\tilde{t}_{d,k,i})
			\!+\! (\tilde{r}_{d,k,i}^{(n)}\!-\!\tilde{t}_{d,k,i}^{(n)})^2 \!-\! 4S_{d,k,i}] ,
			\widetilde{V}_3
			\! \triangleq\!\!
			\sum_{i\in\K}\sum_{k\in\NN}\!
			0.25 [(\tilde{r}_{u,k,j}\!+\!\tilde{t}_{u,k,j})^2\!-\!2(\tilde{r}_{u,k,j}^{(n)}\!-\!\tilde{t}_{u,k,j}^{(n)})(\tilde{r}_{u,k,j}\!-\!\tilde{t}_{u,k,j})
			\!+\! (\tilde{r}_{u,k,j}^{(n)}\!-\!\tilde{t}_{u,k,j}^{(n)})^2 \!-\! 4S_{u,k,j}].$

		Finally, to deal with the constraint \eqref{syncbound-4}, we see that since the function $x\mapsto \frac{1}{x}$ is convex on $(0,+\infty)$, its convex lower bound is obtained by using the first-order Taylor expansion as 
			$\frac{2}{x^{(n)}} - \frac{x}{(x^{(n)})^2} \leq \frac{1}{x},$
		for any $x > 0$, and $x^{(n)}>0$.
		Therefore,   \eqref{syncbound-4} can be replaced by the following convex constraint
		\begin{align}
			\label{syncbound-4-a}
			q_{2,k} \leq LD_kc_k \Big(\frac{2}{f_k^{(n)}} - \frac{f_k}{(f_k^{(n)})^2}\Big), \forall k.
		\end{align}
		Similarly, $V_4(\f, \hat{\ttt}_d, \hat{\ttt}_u)$ has the following convex upper bound:
		\begin{align}
			\nonumber
			\widetilde{V}_4 \!\triangleq\!\! \sum_{k\in\K}\! \Big[t \!-\! \sum_{i\in\K}\hat{t}_{d,k,i} \!-\! LD_kc_k \Big(\frac{2}{f_k^{(n)}} \!-\! \frac{f_k}{(f_k^{(n)})^2}\Big) \!-\! \sum_{i\in\K}\hat{t}_{u,k,i}\Big].
		\end{align}
		
		\begin{algorithm}[!t]
			\caption{Solving problem \eqref{Pmain:epi:relax}}
			\begin{algorithmic}[1]
				\label{alg}
				\STATE \textbf{Initialize}: Set $n\!=\!0$ and choose a random point $\widetilde{\x}^{(0)}\!\in\!\widehat{\FF}$.
				\REPEAT
				\STATE Update $n=n+1$
				\STATE Solve \eqref{Pmain:epi-approx} to obtain its optimal solution $\widetilde{\x}^*$
				\STATE Update $\widetilde{\x}^{(n)}=\widetilde{\x}^*$
				\UNTIL{convergence}
			\end{algorithmic}
		\end{algorithm}
		
		Now, at the iteration $(n+1)$, for a given point $\widetilde{\x}^{(n)}$, problem \eqref{Pmain:epi:relax} can finally be approximated by the following convex problem:
		\begin{align}
			\label{Pmain:epi-approx}
			\underset{\widetilde{\x}\in\widetilde{\FF}}{\min} \,\,
			& \widetilde{\LL} (\widetilde{\x})
		\end{align}
		where $\widetilde{\LL} (\widetilde{\x})
		\triangleq \widetilde{E}_{SB}(\f,\vv_d,\vv_u) \!+\! \lambda( \gamma_1 \widetilde{V}_1(\aaa,\bb) + \gamma_2 \widetilde{V}_2(\tilde{\rr}_d, \tilde{\ttt}_d, \Ss_d) + \gamma_3 \widetilde{V}_3(\tilde{\rr}_u, \tilde{\ttt}_u, \Ss_u) + \gamma_4 \widetilde{V}_4(\f, \hat{\ttt}_d, \hat{\ttt}_u))$ and 
		$\widetilde{\FF}\triangleq\{
		\eqref{sumaki}, \eqref{powerdupperbound}, \eqref{samesizedatad}, 
		\eqref{sumbkj}, \eqref{poweruupperbound}, \eqref{samesizedatau}, 
		\eqref{powerlowerbound}, \eqref{fbound},
		\eqref{power-a-b-relation},
		\eqref{QoSbound-1-a},\\
		\eqref{syncbound-1}-\eqref{syncbound-3},
		\eqref{abrelax},
		\eqref{ratedki-lowerbound-1-a}-
		\eqref{rateukj-upperbound-1-a},
		\eqref{tdki-lowerbound-1}-\eqref{obj-upperbound-2-a}, \eqref{syncbound-4-a}
		\}$ is a convex feasible set.
		In Algorithm~\ref{alg}, we outline the main steps to solve problem \eqref{Pmain:epi:relax}.
		Starting from a random point $\widetilde{\x}\in\widehat{\FF}$, we solve \eqref{Pmain:epi-approx} to obtain its optimal solution $\widetilde{\x}^*$, and use $\widetilde{\x}^*$ as an initial point in the next iteration. The algorithm terminates when an accuracy level of $\varepsilon$ is reached. 
		Algorithm~\ref{alg} converges to a stationary point, i.e., a Fritz John solution, of problem \eqref{Pmain:epi:relax} (hence \eqref{Pmain:epi:equiv} or \eqref{Pmain:sb}). The proof of this fact is rather standard, and it follows from \cite[Proposition 2]{vu18TCOM} and \cite[Proposition 2]{vu18TWC}. 
		
		\subsubsection{Complexity Analysis}
		Problem \eqref{Pmain:epi-approx} can be transformed to an equivalent problem that involves $V_{SB}\triangleq (16K^2+5K+2)$ real-valued scalar variables, $L_{SB}\triangleq (7K^2+12K+1)$ linear constraints, $Q_{SB}\triangleq 12K^2$ quadratic constraints. Therefore, problem \eqref{Pmain:epi-approx} requires a complexity of $\OO(\sqrt{L_{SB}+Q_{SB}}(V_{SB}+L_{SB}+Q_{SB})V_{SB}^2)$ \cite{tam16TWC}. 

		\vspace{-0mm}
		\section{Asynchronous and Synchronous Schemes: \\Problem Formulation and Solution}
		\vspace{-0mm}
		\subsection{Problem Formulation}
		Similarly, the total energy consumption of one FL communication round in the asynchronous and synchronous designs is  
			$E_{Asyn}(\ETA,\ZETA,\f) = E_{Syn} (\ETA,\ZETA,\f) 
			= E (\ETA,\ZETA,\f) \triangleq E_{d}(\ETA) + \sum_{k\in\K} \big( E_{C,k}(f_k) + E_{u,k}(\ZETA) \big)
			= \sum_{k\in\K} \rho_d N_0 \eta_{k} \frac{S_d}{R_{d,k}(\ETA)} +
			\sum_{k\in\K} \!\big( L\frac{\alpha}{2} c_{k}D_kf_{k}^2 + \rho_u N_0\zeta_{k} \frac{S_u}{R_{u,k}(\ZETA)} \big).$
		
		\subsubsection{Optimization Problem for Asynchronous Design} 
		The problem of optimizing power $(\ETA,\ZETA)$ and computing frequency $\f$ to minimize the total energy consumption of one FL communication round in the asynchronous design is formulated as
		\vspace{-4mm}
		\begin{subequations}\label{Pmain:as}
			\begin{align}
				\label{CFas}
				\underset{\ETA,\ZETA,\f}{\min} \,\,
				& E(\ETA,\ZETA,\f)
				\\
				\!\!\!\!\!\mathrm{s.t.}\,\,
				\nonumber
				& \eqref{powerdupperbound:asyn:syn}, \eqref{poweruupperbound:asyn:syn}, \eqref{fbound}
				\\
				\label{powerlowerbound:as}
				& 0\leq \eta_{k}, 0\leq \zeta_{k}, \forall k
				\\
				\label{QoSbound:as}
				& t_{d,k}(\ETA) \!+\! t_{C,k}(f_k) 
				+ t_{u,k}(\ZETA) 
				\leq t_{\text{QoS}}, \forall k
				\\
				\label{syncbound:as}
				& \max_{k\in\K} t_{d,k}(\ETA)
				\leq \min_{k\in\K} \big(t_{d,k}(\ETA) + t_{C,k} (f_k) \big).
			\end{align}
		\end{subequations}
		
		\subsubsection{Optimization Problem for Synchronous Design} Similarly, the problem of optimizing power $(\ETA,\ZETA)$ and computing frequency $\f$ to minimize the total energy consumption of one FL communication round in the synchronous design is formulated as
		\vspace{-0mm}
		\begin{subequations}\label{Pmain:s}
			\begin{align}
				\label{CFs}
				\!\!\!\!\!\underset{\ETA,\ZETA,\f}{\min} \,\,
				& E(\ETA,\ZETA,\f)
				\\
				\!\!\!\!\!\mathrm{s.t.}\,\,
				\nonumber
				& \eqref{powerdupperbound:asyn:syn}, \eqref{poweruupperbound:asyn:syn}, \eqref{fbound}, \eqref{powerlowerbound:as}
				\\
				\label{QoSbound:s}
				& \max_{k\in\K} t_{d,k}(\ETA) \!+\! \max_{k\in\K} t_{C,k}(f_k) 
				+ \max_{k\in\K} t_{u,k}(\ZETA) 
				\leq t_{\text{QoS}}, \forall k.
			\end{align}
		\end{subequations}
		Here, the constraint \eqref{QoSbound:s} captures the nature of ``step-by-step'' scheme, 
		, i.e., every UE needs to wait for all the UEs to finish one step before starting the next step as seen in Fig.~\ref{fig:time1}(a). Compared to \eqref{QoSbound:s}, the constraints \eqref{QoSbound-1} and \eqref{QoSbound:as} provide more flexibility in allocating the available time in Steps (S1)--(S3) to each UE. This is because the UEs in the asynchronous and session-based schemes
		need not wait for other UEs to start a new step.

		\vspace{-0mm}
		\subsection{Solution}
		\vspace{-0mm}
		\subsubsection{Proposed Solution for Asynchronous Design}
		Problem \eqref{Pmain:as} can be transformed into its epigraph form as
		\begin{subequations}\label{Pmain:as:epi}
			\begin{align}
				\label{CFas:epi}
				\!\!\!\!\!\underset{\y}{\min} \,\,
				& \widetilde{E} (\f, \OOmega_d, \OOmega_u) 
				\\
				\!\!\!\!\!\mathrm{s.t.}\,\,
				\nonumber
				& \eqref{powerdupperbound:asyn:syn}, \eqref{poweruupperbound:asyn:syn}, \eqref{fbound}, \eqref{powerlowerbound:as}
				\\
				\label{ratedk-lowerbound-1}
				& r_{d,k} \leq R_{d,k}(\ETA), \forall k
				\\
				\label{rateuk-lowerbound-1}
				& r_{u,k} \leq R_{u,k}(\ZETA), \forall k
				\\
				\label{omegadk}
				& \eta_k - r_{d,k} \omega_{d,k} \leq 0, \forall k
				\\
				\label{omegauk}
				& \zeta_k - r_{u,k} \omega_{u,k} \leq 0, \forall k
				\\
				\label{QoSbound:as:a}
				& \frac{S_d}{r_{d,k}} \!+\! \frac{LD_kc_k}{f_k}
				+ \frac{S_u}{r_{u,k}}
				\leq t_{\text{QoS}}, \forall k
				\\
				\label{syncbound:as:a}
				& \frac{S_{d}}{r_{d,k}} \leq q, \forall k
				\\
				\label{syncbound:as:b}
				& q \leq q_{1,k} + q_{2,k}, \forall k
				\\
				\label{syncbound:as:c}
				& 0 \leq q_{1,k}, 0 \leq q_{2,k}, \forall k
				\\
				\label{syncbound:as:d}
				& q_{1,k} \leq \frac{S_{d}}{r_{d,k}}, \forall k
				\\
				\label{syncbound:as:e}
				& q_{2,k} \leq \frac{LD_kc_{k}}{f_{k}}, \forall k,
			\end{align}
		\end{subequations}
		where $\y \triangleq \{\ETA, \ZETA, \rrr_d, \rrr_u, \f, \OOmega_d, \OOmega_u, q, \pmb{q}_1, \pmb{q}_2\}$, $\rrr_d \triangleq \{r_{d,k}\}, \rrr_u \triangleq \{r_{u,k}\}, \OOmega_d \triangleq \{\omega_{d,k}\}, \OOmega_u \triangleq \{\omega_{u,k}\}, \forall k$, $q, \pmb{q}_1, \pmb{q}_2$ are additional variables, and $\widetilde{E} (\f, \OOmega_d, \OOmega_u) \triangleq \sum_{k\in\K} \rho_d N_0 S_d \omega_{d,k} +
		\sum_{k\in\K} \!\big( L\frac{\alpha}{2} c_{k}D_kf_{k}^2 + \rho_u N_0 S_u \omega_{u,k} \big)$. Problem \eqref{Pmain:as:epi} are still challenging due to nonconvex constraints \eqref{ratedk-lowerbound-1}--\eqref{omegauk}, \eqref{syncbound:as:d}, and \eqref{syncbound:as:e}.
		
		Following the same procedure in Section~\ref{sec:alg:sb}, 
		the concave lower bounds 
		of $R_{d,k}(\ETA)$ and $R_{u,k}(\ZETA)$ is respectively given as
			$\widehat{R}_{d,k} (\ETA)\triangleq \frac{\tau_c - \tau_{d,p}}{\tau_c\log 2} B \Big[ \log\big(1+\frac{\Upsilon^{(n)}}{\Phi^{(n)}}\big) + \frac{2\Upsilon^{(n)}}{(\Upsilon^{(n)}+\Phi^{(n)})} 
			- \frac{(\Upsilon^{(n)})^2}{(\Upsilon^{(n)}+\Phi^{(n)})\Upsilon}
			- \frac{\Upsilon^{(n)}\Phi}{(\Upsilon^{(n)}+\Phi^{(n)})\Phi^{(n)}}\Big],
			\widehat{R}_{u,k} (\ZETA) \triangleq \frac{\tau_c - \tau_{u,p}}{\tau_c\log 2} B \Big[ \log\big(1+\frac{\Psi^{(n)}}{\Omega^{(n)}}\big) + \frac{2\Psi^{(n)}}{(\Psi^{(n)}+\Omega^{(n)})} 
			- \frac{(\Psi^{(n)})^2}{(\Psi^{(n)}+\Omega^{(n)})\Psi}
			- \frac{\Psi^{(n)}\Omega}{(\Psi^{(n)}+\Omega^{(n)})\Omega^{(n)}}\Big], $
		where $\Upsilon(\eta_k) \triangleq (M-K)\rho_d \hat{\sigma}_k^2\eta_{k}$, $\Phi(\ETA) \triangleq \rho_d (\beta_k - \hat{\sigma}_k^2) \sum_{\ell\in\K} \eta_{\ell} +1, \Psi(\zeta_k) \triangleq (M-K) \rho_u \bar{\sigma}_{k}^2 \zeta_{k}$, and $\Omega(\ZETA) \triangleq \rho_u \sum_{\ell \in \K}  (\beta_{\ell} - \bar{\sigma}_{\ell}^2) \zeta_{\ell} + 1$. Then, constraints \eqref{ratedk-lowerbound-1} and  \eqref{rateuk-lowerbound-1} can be approximated by the following convex constraints
		\begin{align}
			\label{ratedk-lowerbound-1-a}
			& r_{d,k} \leq \widehat{R}_{d,k}(\ETA), \forall k
			\\
			\label{rateuk-lowerbound-1-a}
			& r_{u,k} \leq \hat{R}_{u,k}(\ZETA), \forall k.
		\end{align}
		Also, 
		constraints \eqref{omegadk}, and \eqref{omegauk} can be approximated by the following respective convex constraints
		\begin{align}
			\nonumber
			&0.25 [4\eta_k + (r_{d,k}-\omega_{d,k})^2-2(r_{d,k}^{(n)}+\omega_{d,k}^{(n)})(r_{d,k}+\omega_{d,k})
			\\
			& \qquad + (r_{d,k}^{(n)}+\omega_{d,k}^{(n)})^2] \leq 0, \forall k
			\\
			\nonumber
			\label{omegauk:a}
			&0.25 [4\zeta_k + (r_{u,k}-\omega_{u,k})^2-2(r_{u,k}^{(n)}+\omega_{u,k}^{(n)})(r_{u,k}+\omega_{u,k})
			\\
			& \qquad + (r_{u,k}^{(n)}+\omega_{u,k}^{(n)})^2] \leq 0, \forall k.
		\end{align}
		Finally, 
		we replace constraint \eqref{syncbound:as:d} by the following convex constraint
		\begin{align}
			\label{q1:a}
			q_{1,k} \leq S_d \Big(\frac{2}{r_{d,k}^{(n)}} - \frac{r_{d,k}}{(r_{d,k}^{(n)})^2}\Big), \forall k,
		\end{align}
		and constraint \eqref{syncbound:as:e} by \eqref{syncbound-4-a}.
		
		At the iteration $(n+1)$, for a given point $\y^{(n)}$, problem \eqref{Pmain:as:epi} can be approximated by the following convex problem:
		\begin{align}\label{Pmain:as:approx}
			\underset{\y\in\widetilde{\HHH}}{\min} \,\,
			& \widetilde{E}(\f,\OOmega_d,\OOmega_u),
		\end{align}
		where $\widetilde{\HHH}\triangleq\!\{
		\eqref{powerdupperbound:asyn:syn}, \eqref{poweruupperbound:asyn:syn}, \eqref{fbound}, \eqref{powerlowerbound:as}, \eqref{syncbound-4-a},  \eqref{QoSbound:as:a}-\eqref{syncbound:as:c}, \eqref{ratedk-lowerbound-1-a}-\eqref{q1:a}
		\}$ is a convex feasible set. In Algorithm~\ref{alg:as}, we outline the main steps to solve problem \eqref{Pmain:as:epi}.
		Let $\HHH\triangleq \{\eqref{powerdupperbound:asyn:syn}, \eqref{poweruupperbound:asyn:syn}, \eqref{fbound}, \eqref{powerlowerbound:as},\eqref{ratedk-lowerbound-1}-\eqref{syncbound:as:e}\}$ be the feasible set of problem \eqref{Pmain:as:epi}. 
		Starting from a random point $\y\in\HHH$, we solve \eqref{Pmain:as:approx} to obtain its optimal solution $\y^*$, and use $\y^*$ as an initial point in the next iteration. The algorithm terminates when an accuracy level of $\varepsilon$ is reached. Algorithm~\ref{alg:as} converges to a Fritz John solution of \eqref{Pmain:as:epi} (hence \eqref{Pmain:as}) \cite[Proposition 2]{vu18TCOM} and \cite[Proposition 2]{vu18TWC}. Furthermore, in the setting where $\widetilde{\HHH}$ satisfies Slater's constraint qualification condition,
		Algorithm~\ref{alg} converges to a Karush-Kuhn-Tucker solution of \eqref{Pmain:as:epi} (hence \eqref{Pmain:as}) \cite[Theorem 1]{Marks78OR}.
		
		\begin{algorithm}[!t]
			\caption{Solving problem \eqref{Pmain:as:epi}}
			\begin{algorithmic}[1]\label{alg:as}
				\STATE \textbf{Initialize}: Set $n\!=\!0$ and choose a random point $\y^{(0)}\!\in\!\HHH$.
				\REPEAT
				\STATE Update $n=n+1$
				\STATE Solve \eqref{Pmain:as:approx} to obtain its optimal solution $\y^*$
				\STATE Update $\y^{(n)}=\y^*$
				\UNTIL{convergence}
			\end{algorithmic}
			\vspace{+0mm}
		\end{algorithm}
		
		\subsubsection{Proposed Algorithm for Synchronous Design}
		Using a similar procedure to solve problem \eqref{Pmain:as}, we approximate \eqref{Pmain:s} by the following convex problem
		\begin{subequations}\label{Pmain:s:approx}
			\begin{align}
				\underset{\z}{\min} \,\,
				& \widetilde{E}(\f,\OOmega_d,\OOmega_u)
				\\
				\nonumber
				\mathrm{s.t.}\,\,
				& \eqref{powerdupperbound:asyn:syn}, \eqref{poweruupperbound:asyn:syn}, \eqref{fbound}, \eqref{powerlowerbound:as},
				\eqref{ratedk-lowerbound-1-a}-\eqref{omegauk:a}
				\\
				& t_d + t_C + t_u \leq t_{\text{QoS}}
				\\
				& \frac{S_{d}}{r_{d,k}} \leq t_d, \forall k
				\\
				& \frac{LD_kc_{k}}{f_{k}} \leq t_C, \forall k
				\\
				& \frac{S_{u}}{r_{u,k}} \leq t_u, \forall k,
			\end{align}
		\end{subequations}
		where $\rr_d,\rr_u,\OOmega_d,\OOmega_u,t_d,t_C,t_u$ are additional variables and $\z \triangleq \{\ETA, \ZETA, \f,\rr_d,\rr_u,\OOmega_d,\OOmega_u,t_d,t_C,t_u\}$.
		Then, problem \eqref{Pmain:as} can be solved by using Algorithm~\ref{alg:as} where Step 4 solves \eqref{Pmain:s:approx} (instead of \eqref{Pmain:as:approx}).

		\subsubsection{Complexity Analysis}
		\vspace{-0mm}
		The transformed versions of problems \eqref{Pmain:as:approx} and \eqref{Pmain:s:approx}  involve smaller numbers of variables and constraints than the version of problem \eqref{Pmain:epi-approx}, i.e., $V_{Asyn}\triangleq (9K+1)$ and $V_{Syn}\triangleq (7K+4)$ real-valued scalar variables, $L_{Asyn}\triangleq (11K+1)$ and $L_{Syn}\triangleq (7K+2)$ linear constraints, $Q_{Asyn}\triangleq 6K$ and $Q_{Syn}\triangleq 7K$ quadratic constraints. Therefore, problems \eqref{Pmain:as:approx} and \eqref{Pmain:s:approx} respectively require complexities of $\OO(\sqrt{L_{Asyn}+Q_{Asyn}}(V_{Asyn}+L_{Asyn}+Q_{Asyn})V_{Asyn}^2)$ and $\OO(\sqrt{L_{Syn}+Q_{Syn}}(V_{Syn}+L_{Syn}+Q_{Syn})V_{Syn}^2)$, 
		which are lower than that of problem \eqref{Pmain:epi-approx}. 

		\vspace{-0mm}
		\section{Numerical Examples}
		\vspace{-0mm}
		\label{sec:sim}
		\subsection{Network Setup and Parameter Settings}
		\label{sec:sim:parameter}
		\vspace{-0mm}
		We consider an mMIMO network in a square of $D\times D$, where the BS is located at the center and the UEs are located randomly within the square. 
		We choose $D=0.25$ km, and set $\tau_c\!=\!200$ samples.
		The large-scale fading coefficients,  $\beta_{k}$, are modeled in the same manner as \cite{3gpp10}:
			$\beta_k[\text{dB}] = - 148.1  - 37.6 \log_{10}\left(\frac{d_k}{1\,\,\text{km}}\right) + z_k,$
		where $d_k \geq 35$ m is the distance between a UE $k$ and the BS, $z_k$ is a shadow fading coefficient modelled according to a log-normal distribution with zero mean and $7$-dB standard deviation. We choose $B = 20$ MHz, $\tau_{d,p} =\tau_{u,p} \!=\! K$, $S_d\!=\!S_u\!=\!1$ MB, noise power $\sigma_0^2\!=\!-92$ dBm, $L=5$, $f_{\max}=5 \times 10^9$ cycles/s, $D_k = 10^4$ samples, $c_k = 20$ cycles/samples \cite{tran19INFOCOM}, for all $k$, and $\alpha=5\times 10^{-21}$. 
		Let $\tilde{\rho}_d\!=\!10$ W, $\tilde{\rho}_u\!=\!0.2$ W and $\tilde{\rho}_p\!=\!0.2$ W be the maximum transmit power of the BS, UEs and UL pilot sequences, respectively. The maximum transmit powers $\rho_d$, $\rho_u$ and $\rho_p$ are normalized by the noise power.


		\vspace{-0mm}
		\subsection{Results and Discussion}
		\vspace{-0mm}
		As discussed in Remark~\ref{remark:nolearning}, our paper focuses on the communication aspects rather than the learning aspects of the implementation of FL over wireless networks. 
		Therefore, the simulation results of our paper do not include dataset or the learning performance (e.g., convergence speed, training loss, and test accuracy), which is similar to many existing DigComp works in the literature such as \cite{yang21TWC,zeng21TWC,wu22TCOM,viet22TVT,pavlos22CL}.
		
		\subsubsection{Effectiveness of the Proposed Schemes}
		\label{discuss:effectiveness}
		First, we evaluate the convergence behavior of our proposed Algorithms~\ref{alg} and~\ref{alg:as}. Fig.~\ref{Fig:convergence} shows that Algorithm~\ref{alg} converges within $60$ iterations for the session-based scheme, while Algorithm~\ref{alg:as} converges within $30$ iterations for the asynchronous and synchronous schemes. It should be noted that each iteration of Algorithm~\ref{alg} or~\ref{alg:as} involves solving simple convex programs, i.e., \eqref{Pmain:epi-approx}, \eqref{Pmain:as:approx} and \eqref{Pmain:s:approx}. 
		
		\begin{figure}[t!]
			\centering
			\vspace{-5mm}
			\subfigure[Algorithm 1 for the session-based scheme]
			{\includegraphics[width=0.5\textwidth]{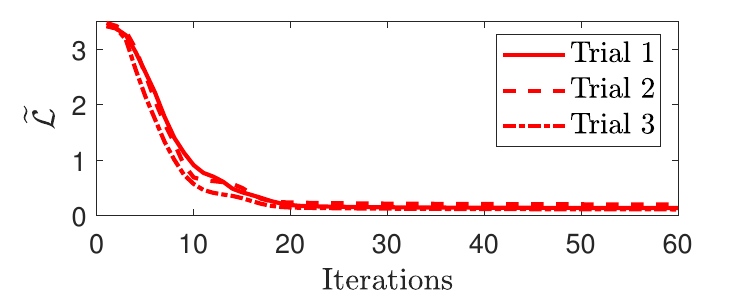}\label{subfig:conv1}}
			\vspace{-2mm}
			\subfigure[Algorithm 2 for the asynchronous scheme]
			{\includegraphics[width=0.5\textwidth]{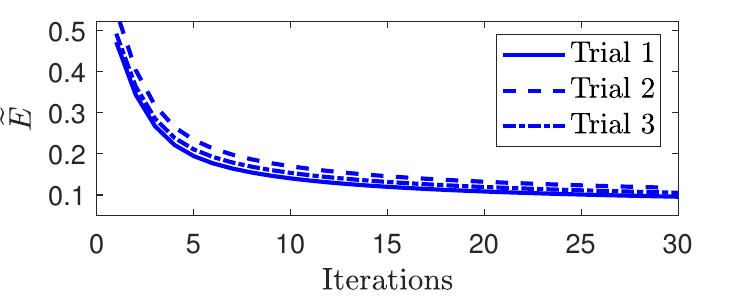}\label{subfig:conv2}}
			\subfigure[Algorithm 2 for the synchronous scheme]
			{\includegraphics[width=0.5\textwidth]{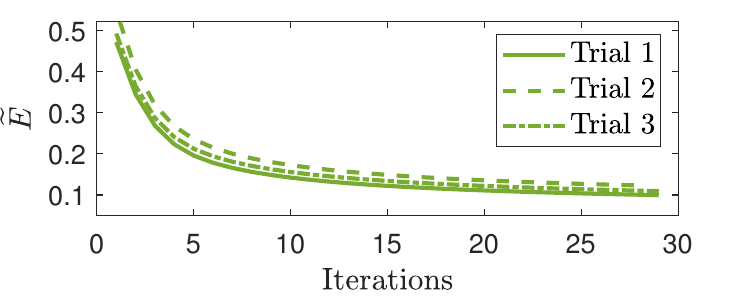}\label{subfig:conv3}}
			\vspace{-2mm}
			\caption{Convergence  of Algorithms~\ref{alg} and~\ref{alg:as}. Here, $M=75, K=10$ and $t_{\text{QoS}}=1$ s.}
			\label{Fig:convergence}
			\vspace{+0mm}
		\end{figure}

		Next, since we are aware of no other existing work that studies energy-efficient massive MIMO networks for supporting FL, we compare the proposed session-based scheme (\textbf{OPT\_SB}), asynchronous scheme (\textbf{OPT\_Asyn}) and synchronous scheme (\textbf{OPT\_Syn}) with the following heuristic schemes:   
		\vspace{-0mm}
		\begin{itemize}
			\item \textbf{HEU\_SB} (Heuristic session-based scheme): In each session, a UE that has a less favorable link condition (i.e., smaller large-scale fading coefficient) is allocated more power to meet the required execution time of one FL communication round. 
			First, since all UEs participated in the DL session $1$ and the UL session $K$, we let $a_{k,1} = b_{k,K} = 1, \forall k$, and take the power allocated to a UE $k$ in the DL session $1$ to be $\eta_{k,1}\!=\!\frac{a_{k,1}(1/\beta_k)}{\sum_{k'\in\K}a_{k',1}(1/\beta_{k'})}$ and the transmit power of a UE $k$ in the UL session $K$ to be $\zeta_{k,K}\!=\!\frac{b_{k,K}(1/\beta_k)}{\sum_{k'\in\K}b_{k',K}(1/\beta_{k'})}$. 
			Now, in each DL session $i, \forall i\neq 1$, since the UE that has the highest data rate finishes its transmission earlier than other UEs and it does not join the subsequent sessions, we choose $a_{k,i}=a_{k,i-1}, k\in\K \setminus k'$, where $k' = \arg\max_{k,i-1} R_{d,k,i-1}$ and $R_{d,k,i-1}$ is obtained by using the given $\ETA_{k,i-1}$ and \eqref{Rd}. Similarly, in each UL session $j, \forall j\neq K$, $b_{k,j-1}=b_{k,j}, k\in\K \setminus k^*$, where $k^* = \arg\max_{k,j} R_{u,k,j}$ and $R_{u,k,j}$ is obtained by using the given $\ZETA_{k,j}$ and \eqref{Ru}. Then, the DL power allocated to a UE $k$ in a DL session $i, \forall i\neq 1$ is $\eta_{k,i}\!=\!\frac{a_{k,i}(1/\beta_k)}{\sum_{k'\in\K}a_{k',i}(1/\beta_{k'})}$ and the UL power of a UE $k$ in an UL session $j, \forall j\neq K$ is $\zeta_{k,j}\!=\!\frac{b_{k,j}(1/\beta_k)}{\sum_{k'\in\K}b_{k',j}(1/\beta_{k'})}$.
			Let $\RR_d \in \RRR^{K\times K}$ and $\RR_u \in \RRR^{K\times K}$, respectively, be the matrices of the DL and UL rates, where $[\RR_d]_{k,i} = R_{d,k,i}$ and $[\RR_u]_{k,j} = R_{u,k,j}$.
			Denote by $\ttt_d\in \RRR^{1\times K}$ and $\ttt_u\in \RRR^{1\times K}$, respectively, be the row vectors comprising the transmission times
			of DL and UL sessions, where $[\ttt_d]_{i} = t_{d,i}$ and $[\ttt_u]_{j} = t_{u,j}$. Let $\pmb{1}\in \RRR^{K\times 1}$ be a column vector of all ones.
			Then, from \eqref{samesizedatad} and \eqref{sametime:sessisionS1}, we have $\ttt_d = (\RR_d^{-1}S_d\pmb{1})^T$. Similarly, from \eqref{samesizedatau} and \eqref{sametime:sessisionS3}, we have $\ttt_u = (\RR_u^{-1}S_u\pmb{1})^T$. The  DL and UL data sizes in each session are calculated according to \eqref{sametime:sessisionS1} and \eqref{sametime:sessisionS3}. 
			The processing frequencies are $f_{k} = \frac{LD_kc_{k}}{t_{\text{QoS}}-\sum_{i\in\K}a_{k,i}t_{d,i} - \sum_{j\in\K}b_{k,j}t_{u,j}}, \forall k$. 
			
			\item \textbf{HEU\_Asyn} (Heuristic asynchronous scheme): The idea of heuristic power allocation in HEU\_SB is applied to the asynchronous scheme. 
			In particular, the DL power to all the UEs are $\eta_{k}\!=\!\frac{(1/\beta_k)}{\sum_{k'\in\K}(1/\beta_{k'})}$ and the UL power of UE $k$ is $\zeta_{k}\!=\!\frac{(1/\beta_k)}{\sum_{k'\in\K}(1/\beta_{k'})}$. The processing frequencies are $f_{k} = \frac{LD_kc_{k}}{t_{\text{QoS}-t_{d,k} - t_{u,k}}}, \forall k$. 
			
			\item \textbf{HEU\_Syn} (Heuristic synchronous scheme):  
			This scheme is similar to HEU\_Asyn, except that the processing frequencies are instead set as $f_{k} \!\!=\!\! \frac{LD_kc_{k}}{t_{\text{QoS}-\max_{k\in\K}t_{d,k} - \max_{k\in\K}t_{u,k}}}$, $\forall k$. 
		\end{itemize} 
		\noindent
		All the following results are obtained by averaging over $200$ channel realizations.
		
		\begin{figure}[t!]
			\centering
			\vspace{-5mm}
			\subfigure[]
			{\includegraphics[width=0.5\textwidth]{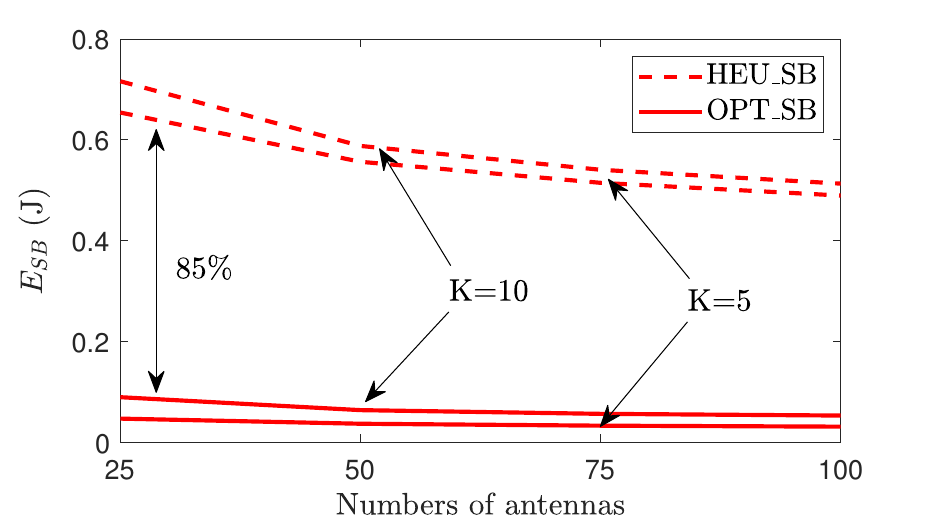}\label{subfig:heu1}}
			\vspace{-2mm}
			\subfigure[]
			{\includegraphics[width=0.5\textwidth]{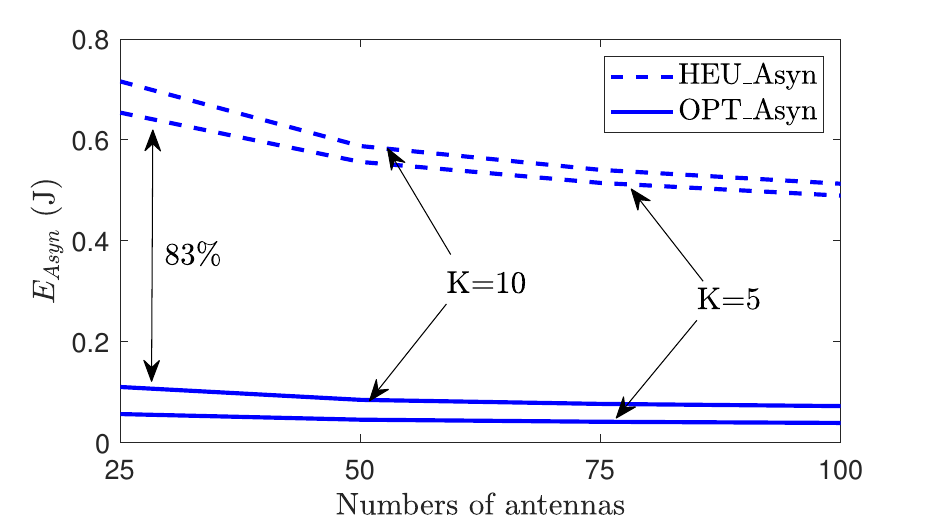}\label{subfig:heu2}}
			\subfigure[]
			{\includegraphics[width=0.5\textwidth]{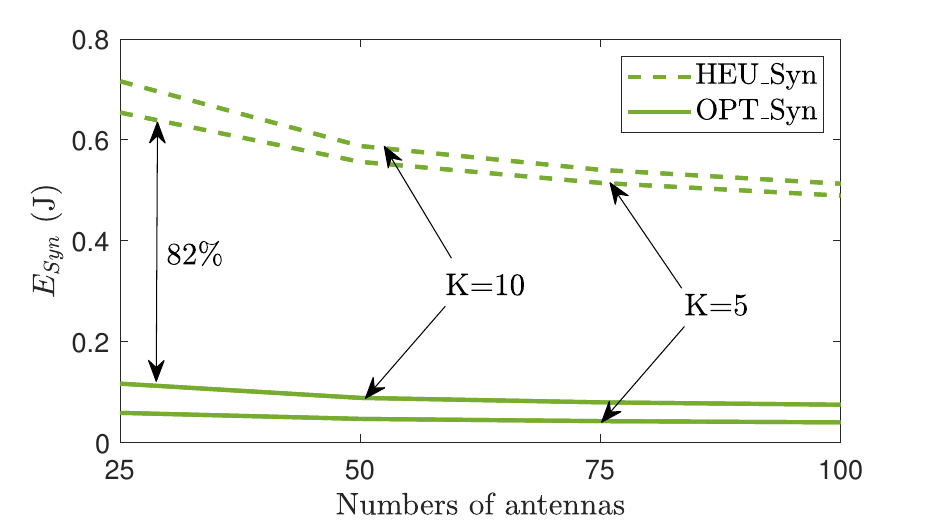}\label{subfig:heu3}}
			\vspace{-4mm}
			\caption{Effectiveness of the proposed schemes with $M=75, K=10$ and $t_{\text{QoS}}=1$ s.}
			\label{Fig:heuvsopt}
		\end{figure}
		
		Fig.~\ref{Fig:heuvsopt} compares the total energy consumption in an FL communication round by all the considered schemes. 
		As seen, our proposed schemes significantly outperform the heuristic schemes. 
		In particular, \textbf{OPT\_SB}, \textbf{OPT\_Asyn}, and \textbf{OPT\_Syn} reduce the total energy consumption by a substantial amount, e.g., by more than $80\%$ in both cases of $K=10$ and $K=5$. 
		These results show the significant advantage of joint optimization of user assignment, data size, time, transmit power, and computing frequencies over the heuristic schemes.
		
		\begin{figure}[t!]
			\centering
			\vspace{-5mm}
			\subfigure[]
			{\includegraphics[width=0.5\textwidth]{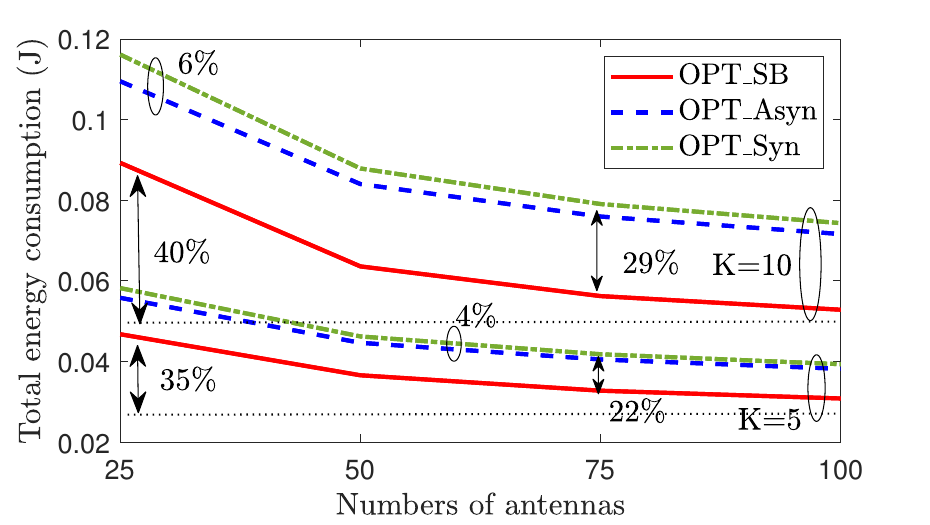}\label{subfig:heu1}\vspace{-0mm}}
			\subfigure[]
			{\includegraphics[width=0.5\textwidth]{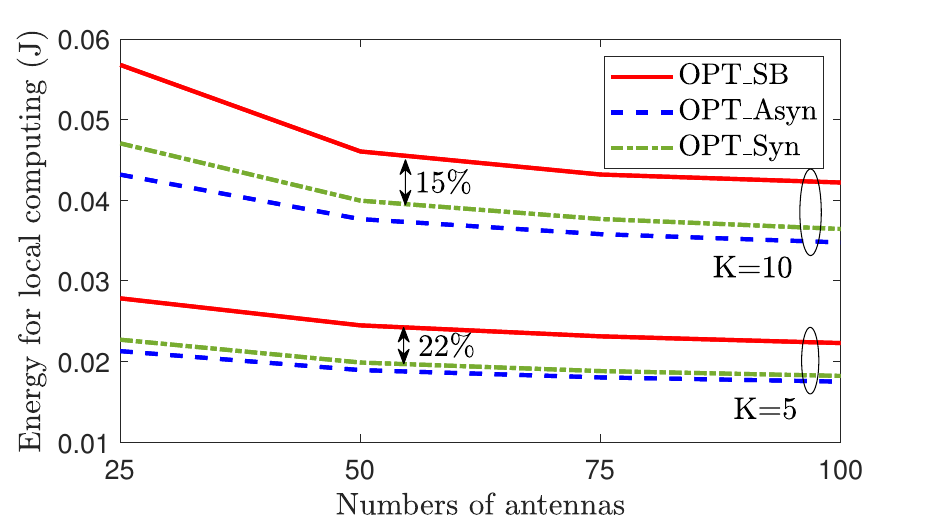}\label{subfig:heu2}}
			\subfigure[]
			{\includegraphics[width=0.5\textwidth]{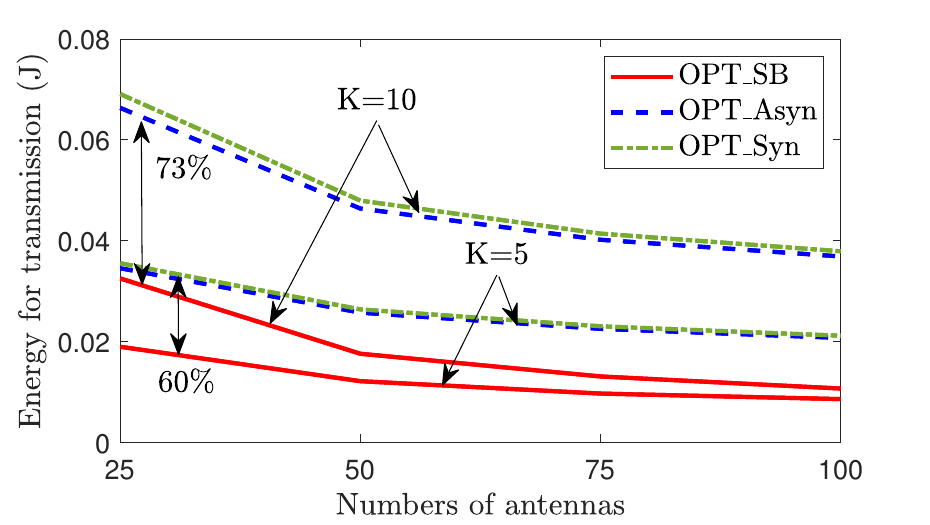}\label{subfig:heu3}}
			\vspace{-4mm}
			\caption{Comparison of the proposed schemes. Here, $t_{\text{QoS}}=1$ s.}
			\label{Fig:Etotal}
			\vspace{+0mm}
		\end{figure}
		
		\subsubsection{Comparison of the Proposed Schemes}
		Fig.~\ref{subfig:heu1} shows that the session-based scheme is the best performer while the synchronous scheme is the worst. Compared to \textbf{OPT\_Syn}, the total energy consumption by \textbf{OPT\_SB} is reduced by up to $29\%$ while that figure for \textbf{OPT\_Asyn} is $6\%$.
		To gain more insights into this result, the total energy consumption for local computing, $E_{C,total}\triangleq \sum_{k\in\K} E_{C,k}(f_k)$, of all considered schemes is shown in Fig.~\ref{subfig:heu2}, and the total energy consumption for transmission $E_{x}-E_{C,total}$ is shown in Fig.~\ref{subfig:heu3}, where $x\in\{SB,Asyn,Syn\}$. 
		First, it can be seen that the energy consumption for local computing and transmission of \textbf{OPT\_Asyn} are both smaller than that of \textbf{OPT\_Syn}. 
		This is so because the UEs in the asynchronous scheme do not wait for other UEs to finish each step. 
		As they have more time available, they can save more energy by using a lower transmit power and a lower computing frequency than the UEs in the synchronous scheme.
		However, the gap between \textbf{OPT\_Asyn} and \textbf{OPT\_Syn} is small because the transmission designs of the asynchronous and synchronous schemes are the same. 
		Here, the session-based scheme uses a more energy-efficient transmission design in which power is not allocated to the UEs who have finished transmission. As a result, compared to the asynchronous and synchronous schemes, the energy consumption for transmission by the session-based scheme is reduced
		by up to $73\%$ as shown in Fig.~\ref{subfig:heu3}. This substantial reduction compensates for the small increase (i.e., $15\%$) in the energy consumption for local computing, making the overall energy consumption by the session-based scheme noticeably lower than that by the asynchronous and synchronous schemes. 
		
		\vspace{-0mm}
		\subsubsection{Impact of the Number of Antennas on the Total Energy Consumption}
		Fig.~\ref{subfig:heu1} also shows that using a large number of antennas corresponds to a reduction of up to $40\%$ in the total energy consumption in one FL communication round. This is because with more antennae, the data rate is higher for the same power level. Thus, the transmission time is shortened, which leads to the reduction in transmission energy; see Fig.~\ref{subfig:heu3}. This also results in more time for local computing, a lower required computing frequency, and then, a reduction in the energy required for local computing as shown in Fig.~\ref{subfig:heu2}. This result shows the importance of massive MIMO technology to support FL.
		
		\subsubsection{Impacts of $t_{\text{QoS}}$ on the Total Energy Consumption of One FL Communication Round}
		
		Fig.~\ref{Fig:tQoS} shows that increasing $t_{\text{QoS}}$
		leads to a dramatic decrease of up to $79\%$ in the total energy consumption. This is because when $t_{\text{QoS}}$ increases, the transmit power and computing frequency required to satisfy the quality-of-service constraint are lower. In turn, they result in a reduction in energy consumption for both transmission and  computing. 
		
		Fig.~\ref{Fig:tQoS} also shows that when increasing $t_{\text{QoS}}$, compared with \textbf{OPT\_Syn}, the energy consumption by \textbf{OPT\_SB} is reduced by
		even more: from $21\%$ for $t_{\text{QoS}}=1$ s to $71\%$ for $t_{\text{QoS}}=4$ s, while the total energy consumption of \textbf{OPT\_Asyn} and \textbf{OPT\_Syn} is almost the same. This result confirms the significant advantage of the session-based transmission design over the conventional transmission designs used in the asynchronous and synchronous schemes.

		\vspace{-0mm}
		\section{Conclusion}
		\vspace{-0mm}
		\label{sec:con}
		In this paper, we  proposed novel synchronous, asynchronous, and  session-based communication designs for massive MIMO networks to support FL. 
		Targeting the minimization of total energy consumption per FL communication round, we formulated design problems 
		that jointly optimize UE assignments, time allocations, transmit powers, and computing frequencies. 
		Relying on successive convex approximation techniques, we  developed novel algorithms to solve the formulated  problems. 
		Numerical results showed that our proposed designs significantly reduced the total energy consumption per FL communication round compared to baseline schemes. In terms of energy savings, the session-based design was the preferred choice to support FL as it outperforms the synchronous and asynchronous designs. 
		For future work, it would be interesting to study the combination of massive MIMO and intelligent reconfigurable surfaces to improve the network coverage as well as taking into account UE selection to improve the energy efficiency of massive MIMO systems to support FL.
		\vspace{-0mm}

		\begin{figure}[t!]
			\centering
			\vspace{-5mm}
			{\includegraphics[width=0.5\textwidth]{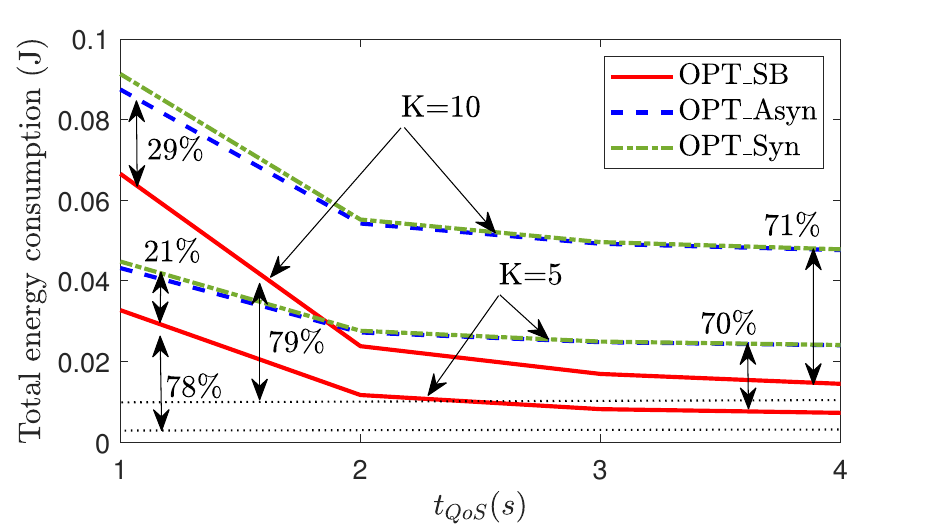}}
			\vspace{-4mm}
			\caption{Impact of $t_{\text{QoS}}$ on the total energy consumption of one FL communication round. Here, $M=75$.}
			\label{Fig:tQoS}
		\end{figure}
		
		\appendix
		\vspace{-1mm}
		Following the arguments in \cite{vu18TCOM,vu18TWC}, let $\widetilde{E}(\lambda)$ 
		be the optimal value 
		at the optimal solution of problem \eqref{Pmain:epi:relax} corresponding to $\lambda$. 
		For ease of presentation, we use $E$ for $\widetilde{E}_{sb}(\f,\vv_d,\vv_u)$. Also, since $(\aaa, \bb, \f, \vv_d, \vv_u, \tilde{\rr}_d, \tilde{\rr}_u,  \Ss_d, \Ss_u, \tilde{\ttt}_d, \tilde{\ttt}_u, \hat{\ttt}_d, \hat{\ttt}_u, \lambda)$ is a subset of variables in $\widetilde{\x}$, we use $\LL(\widetilde{\x},\lambda)$ instead of $\LL(\aaa, \bb, \f, \vv_d, \vv_u, \tilde{\rr}_d, \tilde{\rr}_u,  \Ss_d, \Ss_u, \tilde{\ttt}_d, \tilde{\ttt}_u, \hat{\ttt}_d, \hat{\ttt}_u, \lambda)$. 
		
		Let $\widetilde{E}^*$ be the optimal value of problem \eqref{Pmain:epi:equiv}. Then $\widetilde{E}^* < + \infty$ since $\FF$ is compact. Due to a duality gap between the optimal value of problem \eqref{Pmain:epi:equiv} and the optimal value of its dual problem, we have
			$\sup_{\lambda\geq 0} \widetilde{E}(\lambda)
			= \sup_{\lambda\geq 0} \min_{\widetilde{\x} \in \widehat{\FF}} \LL(\widetilde{\x}, \lambda)
			\leq \widetilde{E}^* \triangleq \min_{\widetilde{\x} \in \widehat{\FF}} \max_{\lambda \geq 0} \LL(\widetilde{\x}, \lambda),$
		which implies that
		\begin{align}
			\label{Estar}
			\widetilde{E}(\lambda) \leq \widetilde{E}^* < +\infty, \forall \lambda\geq 0. 
		\end{align}
		
		\textbf{(i)}  Let $V_{1,\lambda} \triangleq \sum_{k\in\NN} \sum_{i\in\K} ((a_{k,i})_{\lambda}-(a_{k,i})_{\lambda}^2) +  \sum_{k\in\NN} \sum_{j\in\K} ((b_{k,j})_{\lambda}-(b_{k,j})_{\lambda}^2)$, $V_{2,\lambda} \triangleq \sum_{k\in\K} \sum_{i\in\K}  ((\tilde{r}_{d,k,i})_{\lambda} (\tilde{t}_{d,k,i})_{\lambda} - (S_{d,k,i})_{\lambda})$, $V_{3,\lambda} \triangleq \sum_{k\in\K}\sum_{i\in\K} ((\tilde{r}_{u,k,j})_{\lambda} (\tilde{t}_{u,k,j})_{\lambda} - (S_{u,k,j})_{\lambda})$, $V_{4,\lambda} \triangleq \sum_{k\in\K} \big(t_{\lambda} - \sum_{i\in\K}(\hat{t}_{d,k,i})_{\lambda} - t_{C,k}((f_k)_{\lambda}) - \sum_{i\in\K}(\hat{t}_{u,k,i})_{\lambda}\big)$ be the value of $V_1, V_2, V_3, V_4$ at the values 
		$\f_{\lambda}, \aaa_{\lambda}, \bb_{\lambda}, (\tilde{\rr}_d)_{\lambda}, (\tilde{\rr}_u)_{\lambda},  (\Ss_d)_{\lambda}, (\Ss_u)_{\lambda}, (\tilde{\ttt}_d)_{\lambda}, (\tilde{\ttt}_u)_{\lambda}, (\hat{\ttt}_d)_{\lambda}, \\ (\hat{\ttt}_u)_{\lambda}, t_{\lambda}$ corresponding to $\lambda$. Then $V_{1,\lambda}, V_{2,\lambda}, V_{3,\lambda}, V_{4,\lambda} \geq 0, \forall \lambda$. Let $V_{\lambda} \triangleq \gamma_1 V_{1,\lambda} + \gamma_2 V_{2,\lambda} + \gamma_3 V_{3,\lambda} + \gamma_4 V_{4,\lambda}$. Denote by $E_{\lambda}$ the value of $E$ corresponding to $\lambda$.
		Let $0 \leq \lambda_1 < \lambda_2$.
		Because $\widetilde{E}(\lambda_1)$ and $\widetilde{E}(\lambda_2)$ are the optimal values of \eqref{Pmain:epi:relax} corresponding to $\lambda_1$ and $\lambda_2$, we have
		\begin{align}
			\label{Elambda1}
			&\widetilde{E}(\lambda_1) = E_{\lambda_1} + \lambda_1 V_{\lambda_1} \leq  E_{\lambda_2} + \lambda_1 V_{\lambda_2},
			\\
			\label{Elambda2}
			&\widetilde{E}(\lambda_2) = E_{\lambda_2} + \lambda_2 V_{\lambda_2} \leq  E_{\lambda_1} + \lambda_2 V_{\lambda_1}.
		\end{align}
		The above two inequalities lead to $\lambda_1 V_{\lambda_1} + \lambda_2 V_{\lambda_2} \leq \lambda_1 V_{\lambda_2} + \lambda_2 V_{\lambda_1}$, which means $V_{\lambda_2} \leq V_{\lambda_1}$. We conclude that $V_{\lambda}$ is decreasing and bounded below by $0$ as $\lambda$ is increasing, i.e.,
			$V_{\lambda} \rightarrow V^* \geq 0,\,\, \text{as}\,\, \lambda \rightarrow \infty.$
		On the other hand, \eqref{Elambda1} and \eqref{Elambda2} also give
			$\lambda_2 E_{\lambda_1} + \lambda_1 E_{\lambda_2} \leq \lambda_2 E_{\lambda_2} + \lambda_1 E_{\lambda_1},$
		which means $E_{\lambda_2} \geq E_{\lambda_1}$. Thus, $E_{\lambda}$ is increasing and hence bounded below as $\lambda \rightarrow + \infty$. 
		From these observations, if $V^* > 0$, then $\widetilde{E}(\lambda) = E_{\lambda} + \lambda V_{\lambda} \rightarrow + \infty$ as $\lambda \rightarrow + \infty$. This result contradicts \eqref{Estar}, and hence $V^* = 0$. Now, since $V_{\lambda}\rightarrow 0$ as $\lambda \rightarrow + \infty$ and $V_{1,\lambda},V_{2,\lambda},V_{3,\lambda},V_{4,\lambda} \geq 0$, we must have  $V_{1,\lambda} \rightarrow 0, V_{2,\lambda} \rightarrow 0, V_{3,\lambda}\rightarrow 0, V_{4,\lambda}\rightarrow 0$ as $\lambda \rightarrow + \infty$.
		
		\textbf{(ii)}  Denote by $\widetilde{\x}_{\lambda}$ the value of $\widetilde{\x}$ corresponding to $\lambda$. The sequence $\{\widetilde{\x}_{\lambda}\}_{\lambda\geq 0} \subset \widehat{\FF}$ is bounded, and hence, has convergent subsequences. Denote by $\widetilde{\x}_*$ any limit point of $\{\widetilde{\x}_{\lambda}\}_{\lambda}$ as $\lambda \rightarrow + \infty$. Without loss of generality, we assume that $\widetilde{\x}_{\lambda} \rightarrow \widetilde{\x}_*$. Then $\aaa_{\lambda} \rightarrow \aaa_{*}, \bb_{\lambda} \rightarrow \bb_{*}, (\tilde{\rr}_d)_{\lambda} \rightarrow (\tilde{\rr}_d)_{*}, (\tilde{\ttt}_d)_{\lambda} \rightarrow (\tilde{\ttt}_d)_{*}, (\Ss_d)_{\lambda} \rightarrow (\Ss_d)_{*}, (\tilde{\rr}_u)_{\lambda} \rightarrow (\tilde{\rr}_u)_{*}, (\tilde{\ttt}_u)_{\lambda} \rightarrow (\tilde{\ttt}_u)_{*}, (\Ss_u)_{\lambda} \rightarrow (\Ss_u)_{*}, \f_{\lambda} \rightarrow \f_{*}, (\vv_d)_{\lambda} \rightarrow (\vv_d)_{*}, (\vv_u)_{\lambda} \rightarrow (\vv_u)_{*}, (\hat{\ttt}_d)_{\lambda} \rightarrow (\hat{\ttt}_d)_{*}, (\hat{\ttt}_u)_{\lambda} \!\rightarrow\! (\hat{\ttt}_u)_{*}, t_{\lambda} \rightarrow t_{*}$. 
		Thus, $V_{1,\lambda} \rightarrow (V_{1})_* \!\triangleq\! \sum_{k\in\NN} \sum_{i\in\K} ((a_{k,i})_{*}-(a_{k,i})_{*}^2) +  \sum_{k\in\NN} \sum_{j\in\K} ((b_{k,j})_{*}-(b_{k,j})_{*}^2)$, $V_{2,\lambda} \rightarrow (V_{2})_* \!\triangleq\! \sum_{k\in\K} \sum_{i\in\K} ((\tilde{r}_{d,k,i})_{*} (\tilde{t}_{d,k,i})_{*} - (S_{d,k,i})_{*})$, $V_{3,\lambda} \rightarrow (V_{3})_* \triangleq \sum_{k\in\K}\sum_{i\in\K} ((\tilde{r}_{u,k,j})_{*} (\tilde{t}_{u,k,j})_{*} - (S_{u,k,j})_{*}), V_{4,\lambda} \rightarrow (V_{4})_{*} \triangleq \sum_{k\in\K} \big(t_{*} - \sum_{i\in\K}(\hat{t}_{d,k,i})_{*} - t_{C,k}((f_k)_{*}) - \sum_{i\in\K}(\hat{t}_{u,k,i})_{*}\big),
		V_{\lambda} \rightarrow V_* \triangleq \gamma_1 (V_1)_* + \gamma_2 (V_2)_* + \gamma_3 (V_3)_* + \gamma_4 (V_4)_*, E_{\lambda}\! \rightarrow E_{*} \triangleq \widetilde{E}_{sb}(\f_*,(\vv_d)_*,(\vv_u)_*)$. 
		From the results in \textbf{(i)} above, we have $(V_{1})_* = 0, (V_{2})_* = 0, (V_{3})_* = 0, (V_4)_* = 0$, and $V_* = 0$. Thus, $(\aaa_{*}, \bb_{*})$, $((\tilde{\rr}_d)_{*}, (\tilde{\ttt}_d)_{*}, (\Ss_d)_{*}))$, $((\tilde{\rr}_u)_{*}, (\tilde{\ttt}_u)_{*}, (\Ss_u)_{*}))$, and $(\f_*, (\hat{\ttt}_d)_*, (\hat{\ttt}_u)_*, t_*)$ satisfy \eqref{sumab}--\eqref{V4}, respectively.
		Moreover, since  $\widetilde{\x}_{\lambda} \in \widehat{\FF}, \forall \lambda \geq 0$, then $\widetilde{\x}_{*} \in \widehat{\FF}$, which means $\widetilde{\x}_{*} \in \FF$.
		Therefore, $\widetilde{\x}_{*}$ is a feasible point of problem \eqref{Pmain:epi:equiv}, and hence, $E_* \geq \widetilde{E}^*$. By the definition of $\widetilde{E}(\lambda)$, it holds that
			$\sup_{\lambda\geq 0} \widetilde{E}(\lambda) \geq \widetilde{E}(\lambda) = E_{\lambda} + \lambda V_{\lambda} \geq E_{\lambda}, \forall \lambda \geq 0.$
		Letting $\lambda \rightarrow + \infty$ gives
		\begin{align}
			\label{Estar2}
			\sup_{\lambda\geq 0} \widetilde{E}(\lambda) \geq E_* \geq \widetilde{E}^*.
		\end{align}
		From \eqref{Estar} and \eqref{Estar2}, it is true that $\sup_{\lambda\geq 0} \widetilde{E}(\lambda) = E_* = \widetilde{E}^*$, which proves \eqref{Strong:Dualitly:hold}. This implies that $\widetilde{\x}_{*}$ is an optimal solution of \eqref{Pmain:epi:equiv} and the proof is completed. 
		\vspace{-0mm}
		


		\ifCLASSOPTIONcaptionsoff
		\newpage
		\fi

		
		\begin{spacing}{0.99}
			\bibliographystyle{IEEEtran}
			\bibliography{IEEEabrv,newidea2021}
		\end{spacing}
		
		\begin{IEEEbiography}[{\includegraphics*[width=1in, height=1.25in, clip, keepaspectratio]{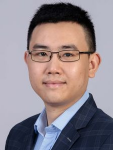}}]
			{Tung Thanh Vu} (Member, IEEE) received the B.Sc. degree (Hons.) in telecommunications and networking from the Ho Chi Minh City University of Science, Vietnam, in 2012, the M.Sc. degree in telecommunication engineering from the Ho Chi Minh City University of Technology, Vietnam, in 2016, and the Ph.D. degree in wireless communications from The University of Newcastle, Australia, in 2021. In 2019, he visited the Broadband Communications Research Lab, McGill University. From February 2021 to September 2022, he was a Research Fellow at Queen's University Belfast, U.K. 
			
			He is currently a Postdoctoral Researcher at the Department of Electrical Engineering (ISY), Linköping University, Sweden. His research interests include optimization, information theories, and machine learning applications for 5G-and-beyond wireless networks, especially with massive MIMO, cell-free massive MIMO, federated learning, full-duplex communications, physical layer security, and low-earth orbit satellite communications. 
			
			Dr. Tung Thanh Vu is serving as an Editor of Elsevier Physical Communication (PHYCOM). 
			He has also served as a member of the technical program committee and the symposium/session chairs in a number of IEEE international conferences such as GLOBECOM, ICCE, ATC. 
			He was an IEEE Wireless Communications Letters exemplary reviewer for 2020 and 2021, an IEEE Transactions on Communications exemplary
			reviewer for 2021. He received the Best Poster Award at AMSI Optimise Conference in 2018. 
		\end{IEEEbiography}
		
		\begin{IEEEbiography}[{\includegraphics*[width=1in, height=1.25in, clip, keepaspectratio]{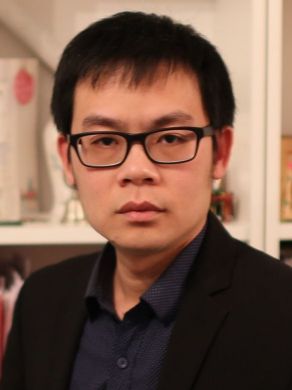}}]
			{Hien Quoc Ngo} (Senior Member, IEEE)  received the B.S. degree in electrical engineering from the Ho Chi Minh City University of Technology, Vietnam, in 2007, the M.S. degree in electronics and radio engineering from Kyung Hee University, South Korea, in 2010, and the Ph.D. degree in communication systems from Link\"oping University (LiU), Sweden, in 2015. In 2014, he visited the Nokia Bell Labs, Murray Hill, New Jersey, USA. From January 2016 to April 2017, Hien Quoc Ngo was a VR researcher at the Department of Electrical Engineering (ISY), LiU. He was also a Visiting Research Fellow at the School of Electronics, Electrical Engineering and Computer Science, Queen's University Belfast, UK, funded by the Swedish Research Council.
			
			Hien Quoc Ngo is currently a Reader (Associate Professor) at Queen's University Belfast, UK. His main research interests include massive (large-scale) MIMO systems, cell-free massive MIMO, physical layer security, and cooperative communications. He has co-authored many research papers in wireless communications and co-authored the Cambridge University Press textbook \emph{Fundamentals of Massive MIMO} (2016).
			
			Dr. Hien Quoc Ngo received the IEEE ComSoc Stephen O. Rice Prize in Communications Theory in 2015, the IEEE ComSoc Leonard G. Abraham Prize in 2017, and the Best PhD Award from EURASIP in 2018. He also received the IEEE Sweden VT-COM-IT Joint Chapter Best Student Journal Paper Award in 2015. He was an \emph{IEEE Communications Letters} exemplary reviewer for 2014, an \emph{IEEE Transactions on Communications} exemplary reviewer for 2015, and an \emph{IEEE Wireless Communications Letters} exemplary reviewer for 2016.  He was awarded the UKRI Future Leaders Fellowship in 2019.
			Dr. Hien Quoc Ngo currently serves as an Editor for the IEEE Transactions on Wireless Communications, the IEEE Wireless Communications Letters, Digital Signal Processing, Elsevier Physical Communication (PHYCOM). He was a Guest Editor of IET Communications, special issue on ``Recent Advances on 5G Communications'' and a Guest Editor of  IEEE Access, special issue on ``Modelling, Analysis, and Design of 5G Ultra-Dense Networks'', in 2017. He has been a member of Technical Program Committees for many IEEE conferences such as ICC, GLOBECOM, WCNC, and VTC.
		\end{IEEEbiography}
		
		\begin{IEEEbiography}[{\includegraphics*[width=1in, height=1.25in, clip, keepaspectratio]{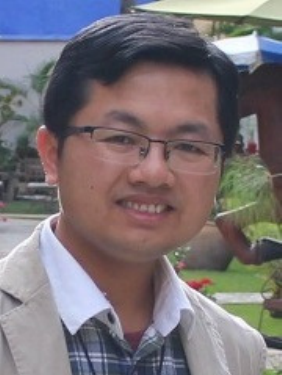}}]
			{Minh Ngoc Dao}  received the B.Sc. (Hons.) and M.Sc. degrees in mathematics from Hanoi National University of Education, Vietnam in 2004 and 2006, respectively, and the Ph.D. degree in applied mathematics from the University of Toulouse, France in 2014. 
			
			Dr. Dao is currently a Senior Lecturer with RMIT University, Australia. Prior to this, he was a Lecturer with Hanoi National University of Education, Vietnam, a Postdoctoral Fellow with The University of British Columbia, Canada, a Research Associate with The University of Newcastle and The University of New South Wales, Australia, and a Lecturer with Federation University Australia. His research interests include mathematical optimization, convex and variational analysis, control theory, signal processing, and machine learning. He received the Annual Best Paper Award from the Journal of Global Optimization in 2017.
		\end{IEEEbiography}
		
		\begin{IEEEbiography}[{\includegraphics*[width=1in, height=1.25in, clip, keepaspectratio]{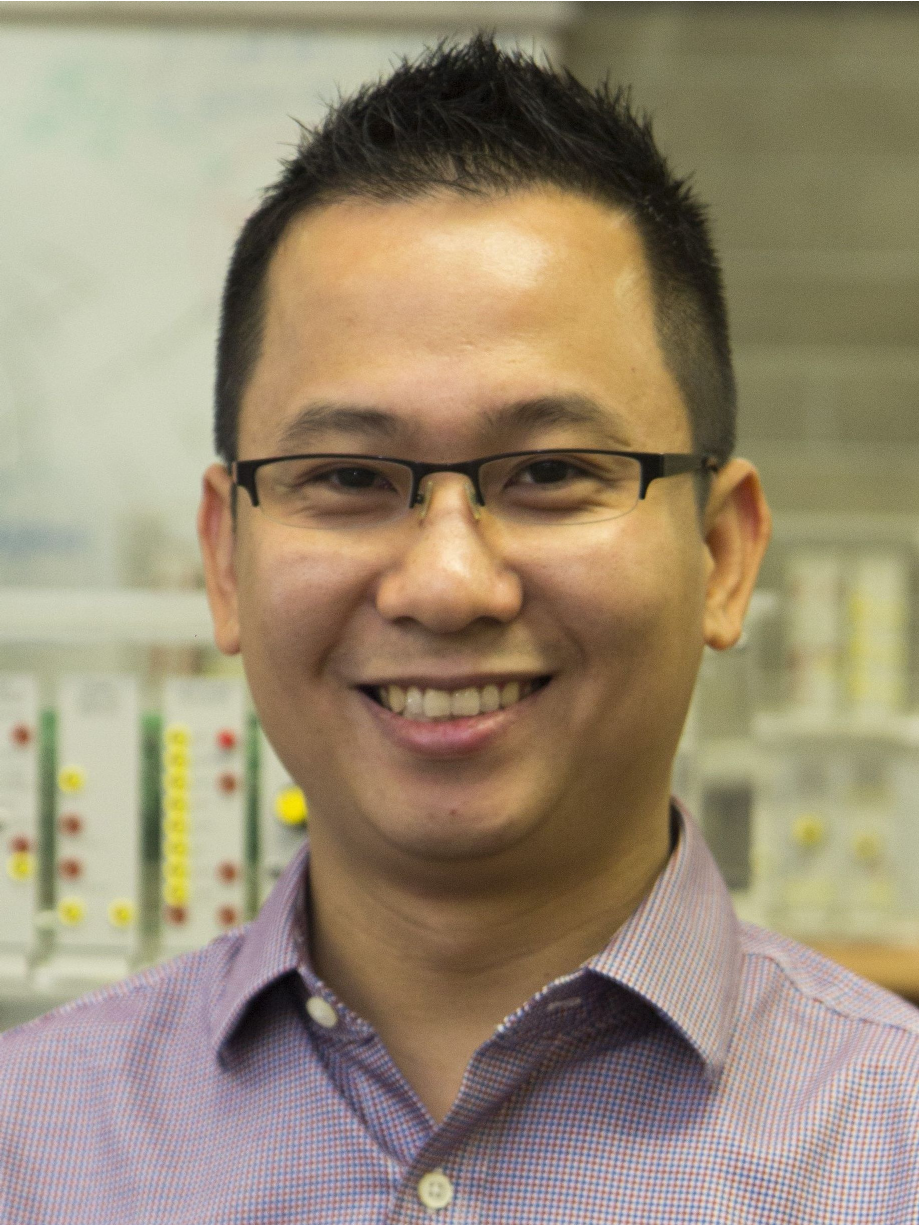}}]
			{Duy Trong Ngo} (Senior Member, IEEE) received the B.Eng. (with First-class Honours and the University Medal) degree in telecommunication engineering from The University of New South Wales in 2007, the M.Sc. degree in electrical engineering (communication) from the University of Alberta in 2009, and the Ph.D. degree in electrical engineering from McGill University in 2013.
			
			Since 2013, Dr. Ngo has been with the University of Newcastle, Australia where he is an Associate Professor in the School of Engineering. His research interests include wireless communications, intelligent transportation systems, and machine learning applications. His research has been funded by Ericsson AB Sweden, the Australian Research Council (Discovery Project, Linkage Project), CRC iMOVE and foreign governments. He was an Associate Editor of the IEEE Wireless Communications Letters during 2020-21.
			
			Dr. Ngo received the NICTA Telecommunications Excellence Award, the University Medal of The University of New South Wales, the Post-Doctoral Fellowships by the Natural Sciences and Engineering Research Council of Canada and the Fonds de recherche du Qu\'{e}bec--Nature et technologies, the Vice-Chancellor's Award for Research and Innovation Excellence, the Pro Vice-Chancellor's Award for Research Excellence, and the Pro Vice-Chancellor's Award for Teaching Excellence.
		\end{IEEEbiography}
		
		\begin{IEEEbiography}[{\includegraphics*[width=1in, height=1.25in, clip, keepaspectratio]{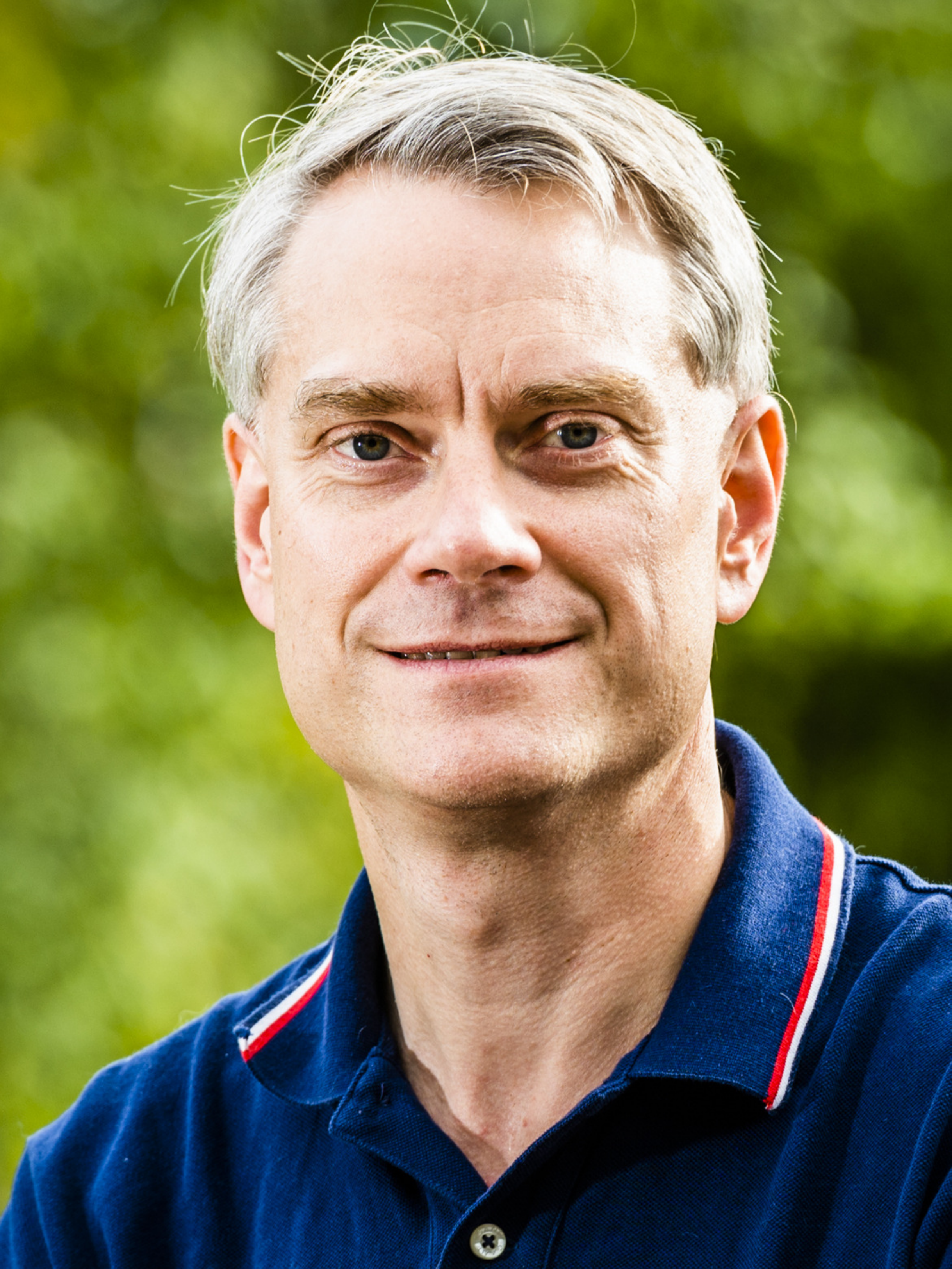}}]
			{Erik G. Larsson} (Fellow, IEEE) received the Ph.D. degree from Uppsala University, Uppsala, Sweden, in 2002.  He is currently Professor of Communication Systems at Link\"oping University (LiU) in Link\"oping, Sweden. He was with the KTH Royal Institute of Technology in Stockholm, Sweden, the George Washington University, USA, the University of Florida, USA, and Ericsson Research, Sweden.  His main professional interests are within the areas of wireless communications and signal processing. He co-authored \emph{Space-Time Block Coding for  Wireless Communications} (Cambridge University Press, 2003) and \emph{Fundamentals of Massive MIMO} (Cambridge University Press, 2016). 
			
			He served as chair  of the IEEE Signal Processing Society SPCOM technical committee (2015--2016), chair of  the \emph{IEEE Wireless  Communications Letters} steering committee (2014--2015), member of the \emph{IEEE Transactions on Wireless Communications} steering committee (2019-2022), General and Technical Chair of the Asilomar SSC conference (2015, 2012), technical co-chair of the IEEE Communication Theory Workshop (2019), and member of the  IEEE Signal Processing Society Awards Board (2017--2019). He was Associate Editor for, among others, the \emph{IEEE Transactions on Communications} (2010-2014), the \emph{IEEE Transactions on Signal Processing} (2006-2010), and the \emph{IEEE Signal  Processing Magazine} (2018-2022).
			
			He received the IEEE Signal Processing Magazine Best Column Award
			twice, in 2012 and 2014, the IEEE ComSoc Stephen O. Rice Prize in
			Communications Theory in 2015, the IEEE ComSoc Leonard G. Abraham
			Prize in 2017, the IEEE ComSoc Best Tutorial Paper Award in 2018, and
			the IEEE ComSoc Fred W. Ellersick Prize in 2019.
		\end{IEEEbiography}
		
		\begin{IEEEbiography}[{\includegraphics*[width=1in, height=1.25in, clip, keepaspectratio]{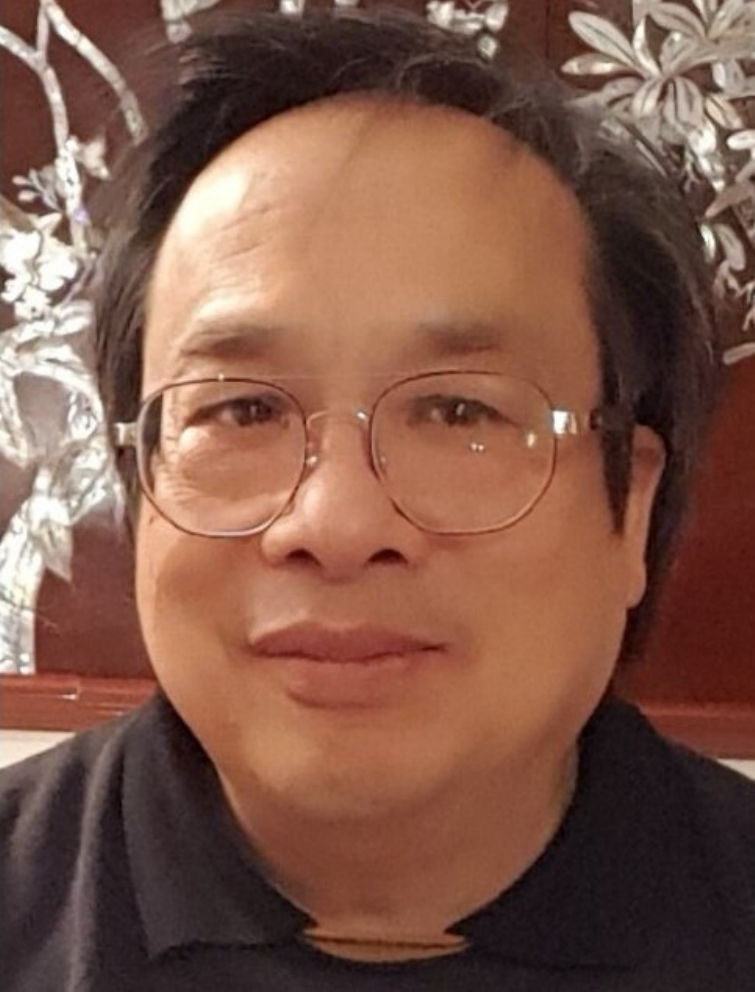}}]
			{Tho Le-Ngoc} (Life Fellow, IEEE) received the B.Eng. degree in electrical engineering, in 1976, the M.Eng. degree in microprocessor applications, in 1978, from McGill University, Montreal, and the Ph.D. degree in digital communications, in 1983, from the University of Ottawa, Canada. From 1977 to 1982, he was with Spar Aerospace Ltd., Sainte-Anne-de-Bellevue, QC, Canada, involved in the development and design of satellite communications systems. From 1982 to 1985, he was with SRTelecom Inc., Saint-Laurent, QC, Canada, where he developed the new point-to-multipoint DA-TDMA/TDM Subscriber Radio System SR500. From 1985 to 2000, he was a Professor with the Department of Electrical and Computer Engineering, Concordia University, Montreal. Since 2000, he has been with the Department of Electrical and Computer Engineering, McGill University. His research interest includes broadband digital communications. 
			
			He is a Distinguished James McGill Professor, and a Fellow of the Engineering Institute of Canada, the Canadian Academy of Engineering, and the Royal Society of Canada. He was a recipient of the 2004 Canadian Award in Telecommunications Research and the IEEE Canada Fessenden Award, in 2005.
		\end{IEEEbiography}

	\end{document}